\documentclass[journal]{IEEEtran}
\usepackage{amsmath}
\usepackage{amssymb}
\usepackage{amsfonts}
\usepackage{graphicx}
\usepackage{epsfig}
\usepackage{subfigure}
\usepackage{psfrag}
\usepackage{color}
\usepackage[noadjust]{cite}
\usepackage{multirow}
\usepackage{epstopdf}

\title{Cognitive Wireless Powered Network: \\ Spectrum Sharing Models and \\ Throughput Maximization
\thanks{The authors are with the Department of Electrical and Computer Engineering, National University of Singapore (email:  s.lee@u.nus.edu, elezhang@nus.edu.sg). R. Zhang is also with the Institute for Infocomm Research, A*STAR, Singapore.} 
}
\author{Seunghyun Lee,~\IEEEmembership{Student Member,~IEEE} and  Rui Zhang,~\IEEEmembership{Senior Member,~IEEE}}

\newtheorem{corollary}{\underline{Corollary}}[section]
\newtheorem{proposition}{\underline{Proposition}}[section]

\def\phi{\varphi}

\def\l{\left}
\def\r{\right}
\def\({\left(}
\def\){\right)}

\def\b0{{\mathbf{0}}}

\def\cF{\mathcal{F}}

\def\cL{\mathcal{L}}

\def\cR{\mathcal{R}}

\newcommand{\nn}{\nonumber}

\def\Pmax{P_\mathsf{max}}

\begin{document}
\maketitle \thispagestyle{empty}
\begin{abstract}
The recent advance in radio-frequency (RF) wireless energy transfer (WET) has motivated the study of \emph{wireless powered communication network} (WPCN), in which distributed wireless devices are powered via dedicated  WET by the hybrid access-point (H-AP) in the downlink (DL) for uplink (UL) wireless information transmission (WIT). In this paper, by utilizing the cognitive radio (CR) technique,  we study a new type of CR enabled secondary WPCN, called cognitive WPCN, under spectrum sharing with the primary wireless communication system. In particular, we consider a cognitive WPCN, consisting of one single H-AP with constant power supply and distributed wireless powered users, shares the same spectrum for its DL WET and UL WIT with an existing primary communication link, where the WPCN's WET/WIT and the primary link's WIT may interfere with each other. Under this new setup, we propose two coexisting models for spectrum sharing of the two systems, namely \emph{underlay} and \emph{overlay} based cognitive WPCNs, depending on different types of knowledge on the primary user transmission available for the cognitive WPCN. For each model, we maximize the sum-throughput of the cognitive WPCN by optimizing its transmission under different constraints applied to protect the primary user transmission. Analysis and simulation results are provided to compare the sum-throughput of the cognitive WPCN versus the achievable rate of the primary user under two proposed coexisting models. It is shown that the overlay based cognitive WPCN outperforms the underlay based counterpart, thanks to its fully cooperative WET/WIT design with the primary WIT, while also requiring higher complexity for implementation.        

\end{abstract}


\section{Introduction} \label{Section:Introduction}
\PARstart{R}ECENTLY, wireless powered communication network (WPCN) enabled by radio-frequency (RF) wireless energy transfer (WET) has emerged as a promising new area of research in wireless communication (see, e.g., \cite{J_BHZ:2015} and the references therein). A typical model of WPCN  is illustrated in Fig.~\ref{Fig:GenericModel} \cite{J_JZ:2014}, where a WET-enabled hybrid access point (H-AP) with constant power supply coordinates the downlink (DL) WET and uplink (UL) wireless information transmission (WIT) to/from a set of distributed user terminals. A practical transmission protocol for the WPCN  is the so-called ``harvest-then-transmit" protocol \cite{J_JZ:2014}, as shown in Fig.~\ref{Fig:HarvestThenTransmit}, where in each transmission block of duration $T$,\footnote{For convenience, we assume normalized block time with $T=1$ in the sequel of this paper.} all users first harvest energy from the DL WET by the H-AP during the \emph{WET phase} with the first $\tau$ amount of time, and then send their information to the H-AP during the \emph{WIT phase} of remaining time $T-\tau$ using their individually harvested energy. By optimally coordinating and scheduling the WET and WIT in a WPCN, the user terminals can be free from energy depletion, even without any fixed energy sources, and achieve more reliable throughput than conventional wireless networks with battery powered devices where the battery needs to be manually replaced/recharged after depletion. Besides, different from the existing technique of harvesting energy from ambient RF transmission that is non-intended for energy transfer, wireless devices in a WPCN are powered by dedicated WET that is fully controllable in waveform design, power level, and time/spectrum allocation, thus providing far more reliable and substantial power supplies. It is worth noting that current WET technology can already deliver tens of microwatts RF power to wireless devices from a distance of more than 10 meters,\footnote{Please refer to the company website of Powercast Corp. (http://www.powercastco.com) for more information on RF-based WET.} and thus WPCN is suitable for low-power applications such as wireless sensor networks and RF identification (RFID) networks where device operating power is typically low. 

\begin{figure} 
\centering
\includegraphics[width=8cm]{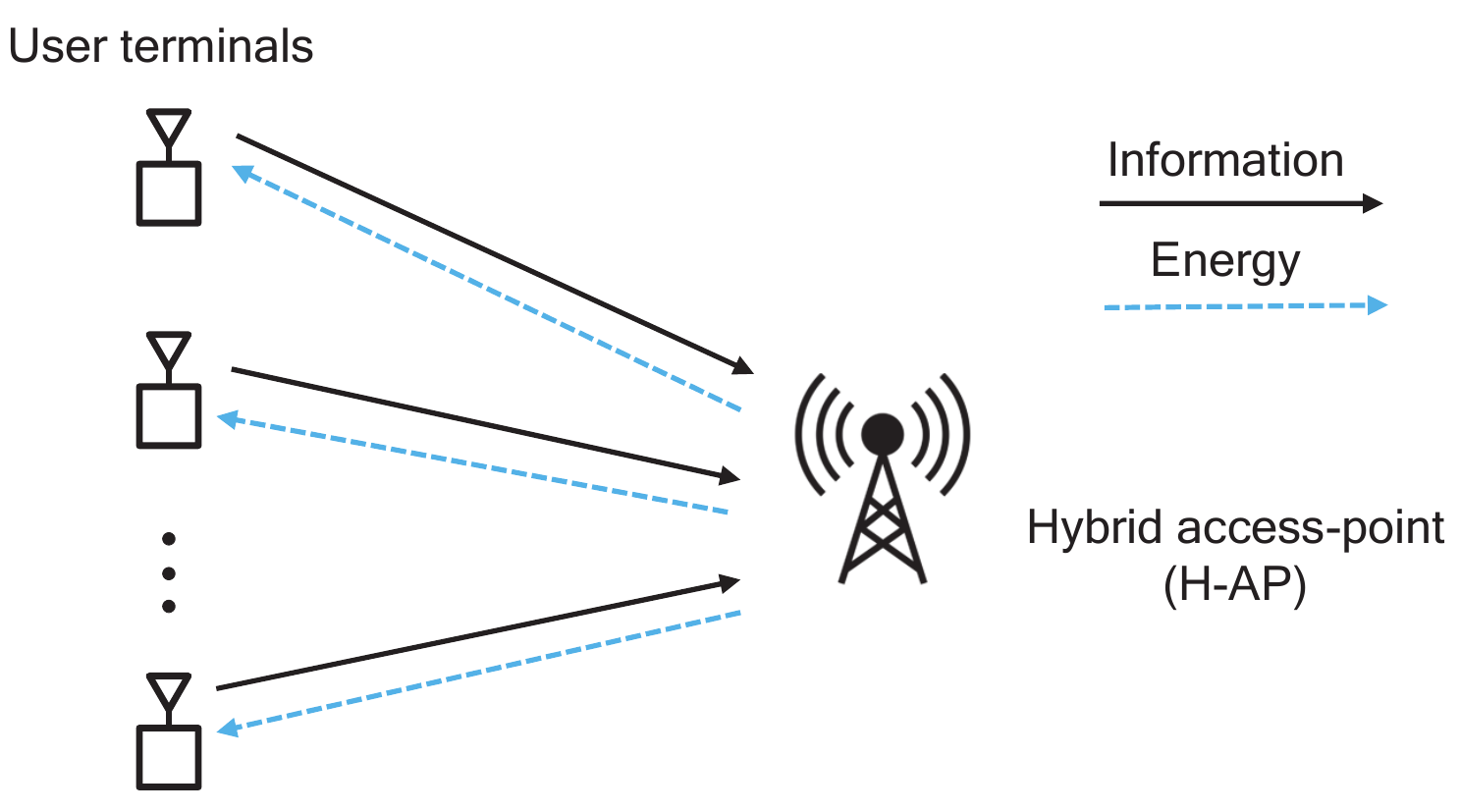}
\caption{A WPCN with downlink wireless energy transfer (WET) and uplink  wireless information transmission (WIT).  } \label{Fig:GenericModel}
\end{figure}
\begin{figure} 
\centering
\includegraphics[width=8cm]{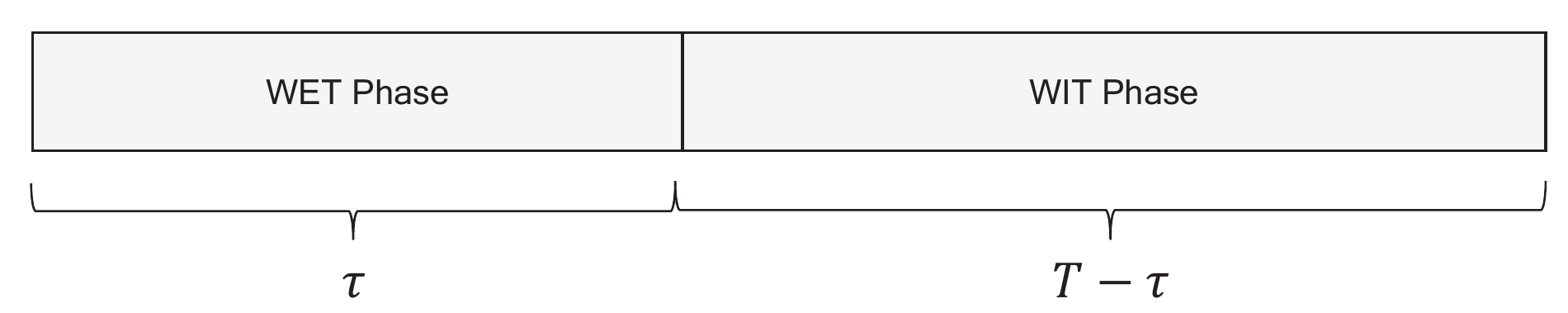}
\caption{The harvest-then-transmit protocol for WPCN \cite{J_JZ:2014}.} \label{Fig:HarvestThenTransmit}
\end{figure}


\subsection{Prior Works}

As a promising solution to achieve perpetually operating wireless networks,  WPCN has been thoroughly investigated in the literature recently (see, e.g., \cite{J_JZ:2014, J_JZ:2014b, C_JZ:2014, J_CLRUV:2015, J_LZC:2014, J_YHZG:2015, J_HL:2014, J_CDZ:2015}). 
Based on the harvest-then-transmit protocol, the sum-throughput maximization in a WPCN with time-division-multiple-access (TDMA) based WIT is first studied in \cite{J_JZ:2014}. This work is then extended to the setup with a full-duplex H-AP that is able to transmit energy and receive information simultaneously over the same band in \cite{J_JZ:2014b}. WPCN with cooperative user communications is also studied in  \cite{C_JZ:2014, J_CLRUV:2015}, where the users with more harvested energy than others can help relay their information to the H-AP, thus achieving more efficient energy utilization in the system. In addition,
 \cite{J_LZC:2014} extends the study in \cite{J_JZ:2014} to a multi-antenna H-AP setup, where more efficient DL WET and UL WIT are achieved via energy beamforming and space-division-multiple-access (SDMA) based transmission, respectively. More recently, \cite{J_YHZG:2015} studies the WPCN with a large number of antennas at the H-AP to significantly enhance the WET/WIT efficiency, by exploiting the ``massive multiple-input multiple-output (MIMO)" gains. Besides, spatial capacity of large-scale WPCNs is studied in \cite{J_HL:2014, J_CDZ:2015}, based on the tools from stochastic geometry. It is also worth pointing out that since RF signals can convey both information and energy at the same time, another appealing line of related research is to explore a dual use of RF signals to optimize the performance of \emph{simultaneous wireless information and power transfer} (SWIPT)  \cite{J_ZH:2013,J_ZZH:2013}.  Various models for different SWIPT applications have been studied to characterize the performance trade-offs between WIT and WET with the same transmitted signal, including e.g., broadcast channel \cite{J_ZH:2013,J_XLZ:2014}, interference channel  \cite{J_PC:2013, J_LLZ:2015}, and relay channel \cite{J_DKSP:2014}.

\begin{figure} 
\centering
\includegraphics[width=8.5cm]{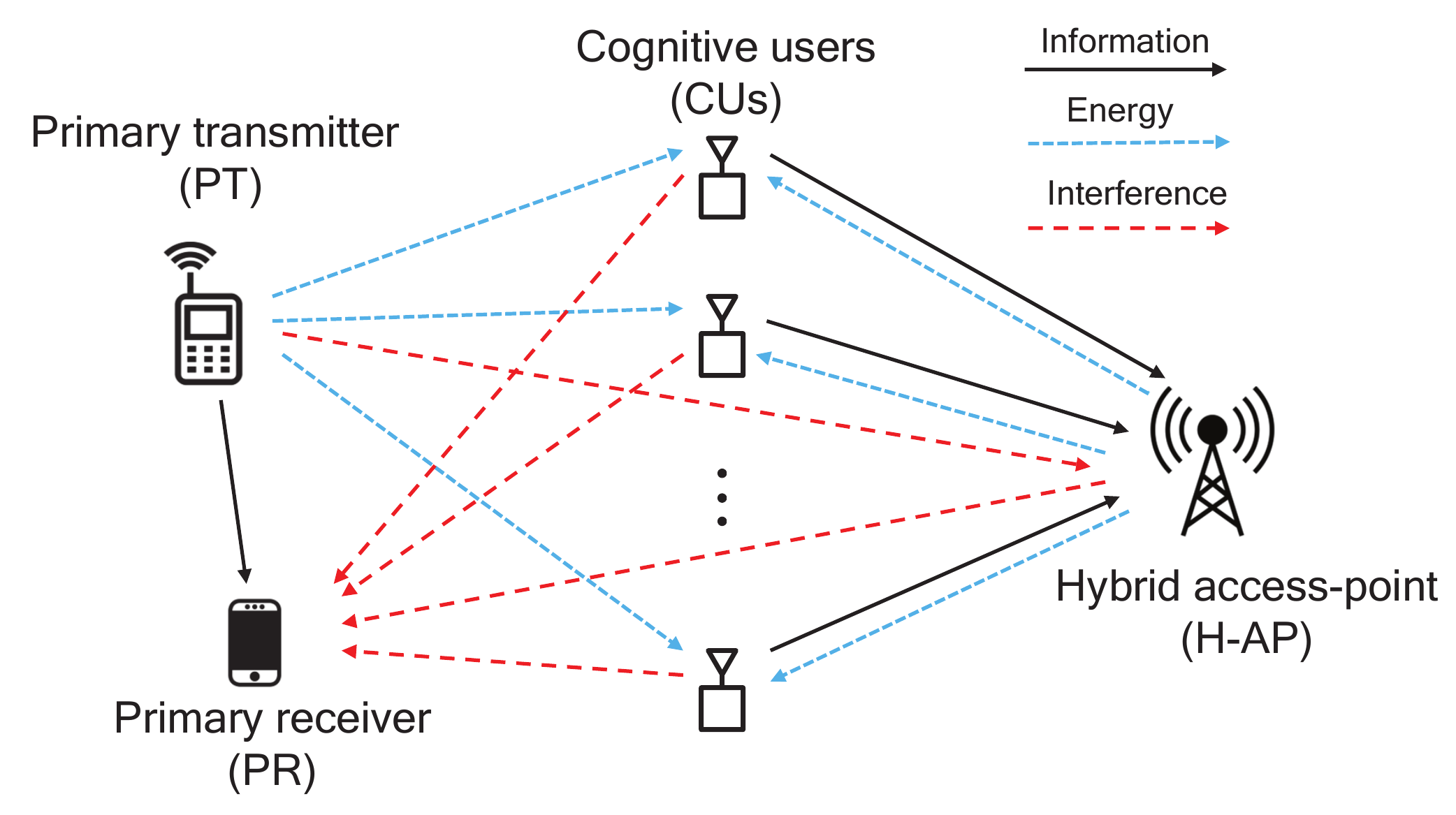}
\caption{Schematic of a secondary cognitive WPCN coexisting with a primary communication link.} \label{Fig:SystemModel}
\end{figure}

\subsection{Motivation and Proposed Models}
However, the above works on WPCN are all based on the assumption that the system can operate in a given frequency band for exclusive use, and thus there is no need to consider the effects of the interference with other wireless systems operating in the same band. Since the available spectrum is currently highly limited due to the ever-expanding wireless systems and applications, finding new available spectrum for future applications such as WPCN will be a challenging task. One promising solution to this issue can be obtained by leveraging the celebrated cognitive radio (CR) technique that can efficiently enable a secondary system to reuse the spectrum of a primary system in an opportunistic manner (see, e.g., \cite{J_ZS:2007, J_GJMS:2009, J_ZLC:2010} and the references therein). This thus motivates our joint investigation of the WPCN and CR based network design in this paper to efficiently implement WPCN without the need of any new spectrum allocated. Specifically, as shown in Fig.~\ref{Fig:SystemModel}, we consider spectrum sharing between a primary WIT link consisting of a single pair of primary transmitter (PT) and primary receiver (PR), and a secondary WPCN consisting of a set of distributed cognitive users (CUs) and one single H-AP, which coordinates the WET/WIT in the WPCN to maximize the throughput of CUs while maintaining the required quality of service for the primary transmission. The CR enabled secondary WPCN is thus referred to as \emph{cognitive WPCN} (CWPCN). As compared to conventional CR networks, the hybrid WET/WIT in CWPCN brings new challenges to the design of optimal transmission. On one hand, the interference from the CWPCN to the primary transmission can be potentially more severe due to the WET in addition to WIT than in conventional CR network. On the other hand, compared to a ``stand-alone" WPCN as in \cite{J_JZ:2014}, the CUs in CWPCN can harvest more energy from the primary interference signal in addition to the dedicated WET from the H-AP. It is worth noting that CR networks powered by RF energy harvesting from the nearby primary transmitters have been studied in e.g., \cite{J_LZH:2013, J_HNWK:2014}. However, different from \cite{J_LZH:2013, J_HNWK:2014}, where CUs scavenge energy from primary transmission only, in our studied CWPCN the CUs are mainly powered via dedicated WET by the H-AP, although they can also be opportunistically powered by the PT transmission, as shown in Fig.~\ref{Fig:SystemModel}. 

\begin{figure}
\centering
\subfigure[U-CWPCN]{
\centering
\includegraphics[width=8.5cm]{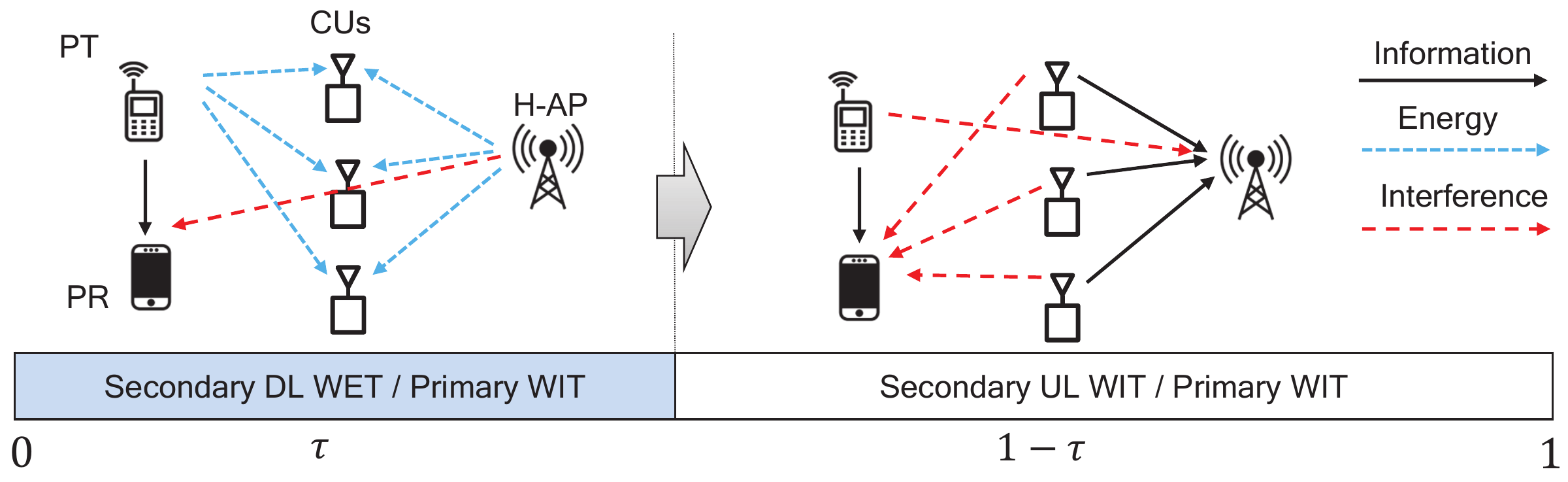}\label{Fig:Underlay}} 
\subfigure[O-CWPCN]{
\centering
\includegraphics[width=8.5cm]{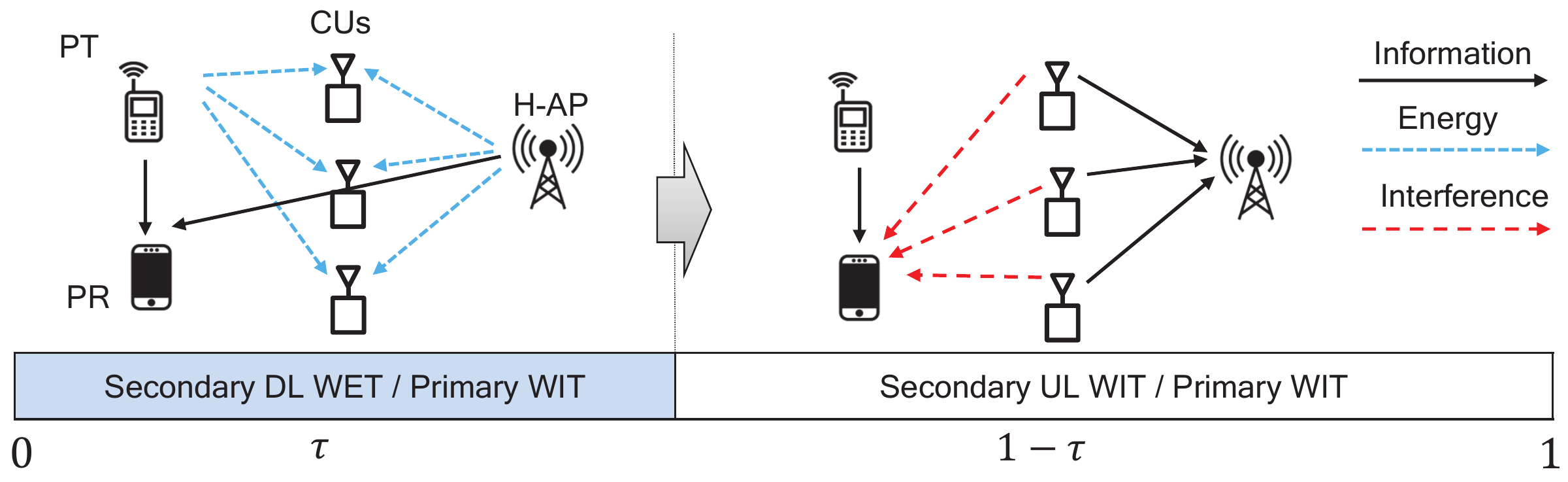}\label{Fig:Overlay}} 
\subfigure[Interweave based CWPCN]{
\centering
\includegraphics[width=8.5cm]{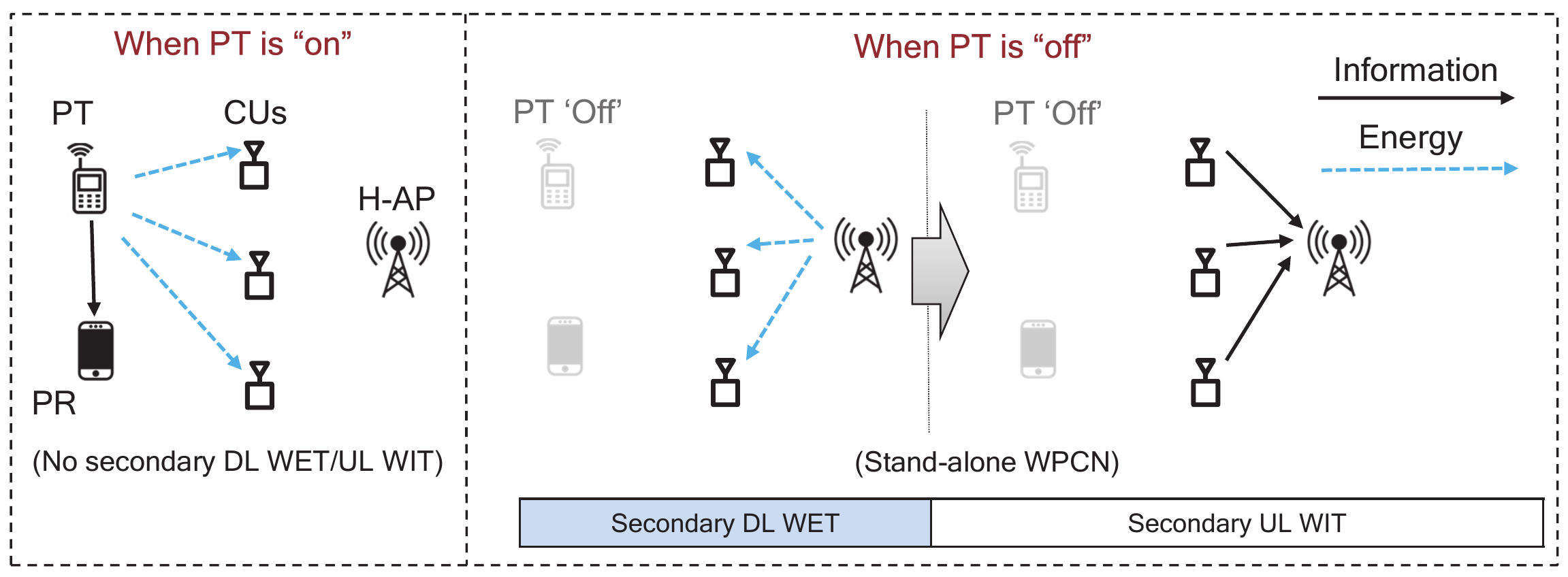}\label{Fig:Interweave}} 
\caption{Comparison of U-CWPCN, O-CWPCN, and interweave based CWPCN.}
\label{Fig:}
\end{figure}

Under the new CWPCN setup, in this paper we propose two coexisting models for spectrum sharing, based on two existing models for  conventional CR networks \cite{J_GJMS:2009}, namely  \emph{underlay} and \emph{overlay} based CWPCNs (U-CWPCN and O-CWPCN), respectively, depending on different types of  knowledge on the primary transmission known by the CWPCN.\footnote{Note that \cite{J_ZS:2007} and \cite{J_GJMS:2009} use different definitions of the terms \emph{underlay} and \emph{overlay} CRs. In this paper, we follow the definitions in \cite{J_GJMS:2009} for the underlay and overlay CRs.} First, consider U-CWPCN as depicted in Fig.~\ref{Fig:Underlay}, which operates under the harvest-then-transmit protocol. Assuming the knowledge of the channel gains from the H-AP/CUs to the PR, both the H-AP's DL WET and the CUs' UL WIT can be jointly designed to maximize the throughput of CUs while meeting the so-called \emph{interference-temperature constraint} (ITC) \cite{J_ZLC:2010}, to ensure that the resulting interference levels at the PR from both the H-AP's DL WET and CUs' UL WIT do not exceed a predefined threshold in the WET phase and WIT phase, respectively. In contrast, for O-CWPCN, it is assumed that the H-AP knows additionally the channels from the PT to the PR/H-AP as well as the PT's messages to be sent to the PR, similarly as in \cite{J_GJMS:2009, J_DMT:2006, J_JV:2009}. With such knowledge, as illustrated in Fig.~\ref{Fig:Overlay}, during the WET phase, the H-AP can  cooperatively send the PT's message using its WET signal to the PR; while in the WIT phase, the H-AP is able to cancel the interference from the PT before decoding the CUs' messages. In this model, the CWPCN's transmission can be designed by applying a more effective \emph{primary rate constraint} (PRC) than the simple ITC in U-CWPCN thanks to the additional channel knowledge on the primary link, where the achievable rate of the primary transmission is guaranteed to be no smaller than a predefined target. 

Notice that in the above two spectrum sharing models, we have assumed that the primary transmission is continuous or always on.  However, if the PT transmits only for a fraction of time and this ``on" time period can be perfectly detected by the CWPCN,\footnote{In this interweave CR based WPCN, an important problem is to detect the coexisting primary user's on/off states via spectrum sensing. Advanced techniques such as cooperative spectrum sensing (see, e.g., \cite{J_ZML:2009} and the references therein) can be used to considerably improve the sensing reliability. Specifically, in our studied system, each of the CUs observes the primary transmission via e.g., energy detection and then forwards its sensing result to the H-AP where a joint decision on whether the PT is active or not is made.  } the CWPCN can work as a stand-alone WPCN when the PT is off, and when the PT is on, the CUs can harvest additional energy from the PT's WIT  whereas the H-AP keeps silent in order not to interfere with the primary transmission, as illustrated in Fig.~\ref{Fig:Interweave}. In this case, given the fraction of time during which the PT is off, the transmit optimization for a single CWPCN has been studied in \cite{J_JZ:2014} without coexisting primary users. Therefore, we omit our discussion on such ``interweave" CR \cite{J_GJMS:2009} based WPCN in this paper, and focus our study on the more challenging U-CWPCN and O-CWPCN models with the coexisting PT that transmits continuously. 

\subsection{Main Results and Organization}
The main results of this paper are summarized as follows.
\begin{itemize}
\item 
For each of the two coexisting models proposed, we formulate the sum-throughput maximization problem for the CWPCN by jointly optimizing the time allocation for the harvest-then-transmit protocol and the WET/WIT transmit power in the CWPCN, under different criteria used to protect the primary  transmission. 

\item
Specifically, for the case of U-CWPCN under a given ITC at the PR in both the WET and WIT phases, we propose an efficient algorithm to solve the CWPCN sum-throughput maximization problem by separately optimizing the WET/WIT time allocation and transmit power in the WIT phase.

\item
While for the case of O-CWPCN, we study the sum-throughput maximization problem under a given PRC instead of ITC in the case of U-CWPCN. The problem is shown to be non-convex. Nonetheless, we solve the problem optimally by solving a sequence of time/power allocation problems in the CWPCN, each subject to a given ITC, and iteratively searching for the optimal ITC level to maximize the sum-throughput of the CWPCN. Furthermore, we show that each aforementioned subproblem with a given ITC can be solved by leveraging the solution obtained for the sum-throughput maximization problem in the U-CWPCN case.

\item Finally, we compare the performance of the two coexisting models in terms of the primary link's achievable rate versus CWPCN's sum-throughput trade-off. We show analytically that the achievable rate region of the O-CWPCN is always no smaller than that of the U-CWPCN, thanks to its fully cooperative WET/WIT design with the primary WIT. Extensive simulation results under various setups are also provided to compare the performance in practical systems. 

\end{itemize}

The rest of this paper is organized as follows. Section~\ref{Section:SystemModel} introduces the system model and the problem formulation for the proposed U-CWPCN and O-CWPCN models, and compares the performance of the two models analytically. Sections~\ref{Section:Underlay} and \ref{Section:Overlay} present the solutions of the sum-throughput maximization problems for each of the two models, respectively. Section~\ref{Section:Simulation} provides simulation results. Section~\ref{Section:Conclusion} concludes the paper and discusses future work. 

Some abbreviations used in this paper are summarized in Table I.

\begin{table} 
\caption{List of Abbreviations}
\centering
\begin{tabular}{l l} 

\hline
RF & Radio-frequency \\
WET &  Wireless energy transfer \\
WIT & Wireless information transmission \\
DL & Downlink \\
UL & Uplink \\
CR & Cognitive radio \\
WPCN & Wireless powered communication network \\
CWPCN & Cognitive WPCN \\
U-CWPCN & Underlay based CWPCN \\
O-CWPCN & Overlay based CWPCN \\
H-AP & Hybrid access-point \\
CU & Cognitive user \\
PT & Primary transmitter \\
PR & Primary receiver \\ 
ITC & Interference-temperature constraint \\
PRC & Primary rate constraint \\
CSI & Channel state information \\

\hline
\end{tabular}
\label{Table:Notation}
\end{table}

\section{System Model and Problem Formulation} \label{Section:SystemModel}

As shown in Fig.~\ref{Fig:SystemModel}, this paper studies a simplified CWPCN, in which a secondary WPCN consisting of $K$ CUs and a single H-AP shares the same bandwidth for WET and WIT with a primary WIT link consisting of a single pair of PT and PR. It is assumed that all terminals are each equipped with a single antenna, and operate in half-duplex mode. In addition, all CUs are assumed to have no fixed power supplies, and thus need to replenish energy by harvesting RF energy from the H-AP's DL WET and/or the PT's WIT (to the PR); the harvested energy at each CU is stored in a rechargeable battery and then used for its  UL WIT. Unlike CUs, it is assumed that the PT and the H-AP have constant power supplies.

We assume a quasi-static flat-fading channel model, where all the channels involved in the network remain constant during each block transmission time with normalized duration $T=1$,\footnote{Hence, in this paper we use the terms of energy and power interchangeably.} but can vary from block to block. Furthermore, we assume coherent communication for all links and thus only the channel power gain needs to be considered. Denote $\hat{h}$, $h_i$, and $\hat{h}_i$, $i=1,...,K$, as the channel power gains from the H-AP of CWPCN to PR, the H-AP to CU$_i$ (and vice versa, by assuming the channel reciprocity), and CU$_i$ to PR, respectively. Similarly, $g$, $g_i$, $i=1,...,K$, and $\hat{g}$  are defined as the channel power gains from the PT to PR, CU$_i$, and H-AP, respectively.  

For the secondary CWPCN, in this paper we adopt the ``harvest-then-transmit" protocol proposed in \cite{J_JZ:2014}, as shown in Fig.~\ref{Fig:HarvestThenTransmit}. Thus, in each transmission block, the transmission in the CWPCN is divided into two phases. In the first \emph{WET phase} of duration $\tau$, each CU performs energy harvesting, while in the subsequent \emph{WIT phase} of duration $1-\tau$, all CUs transmit independent information to the H-AP simultaneously using their individually harvested energy. 

In the following subsections, we introduce the two coexisting models for CWPCN, namely U-CWPCN and O-CWPCN, respectively, and formulate the sum-throughput maximization problem for the CWPCN in each case. Then, we analytically compare the rate performance of the two models.

\subsection{Underlay Based CWPCN (U-CWPCN)} \label{Subsection:Underlay}
First, consider the U-CWPCN depicted in Fig.~\ref{Fig:Underlay}. In the WET phase with the first $\tau$ amount of time, where the H-AP broadcasts signal to all the CUs for DL WET, the total harvested power at CU$_i$ from both the H-AP and the PT can be expressed as (considering independent transmitted signals of H-AP and PT)
\begin{equation} \label{Eq:Energy}
Q_i = \eta_i(h_i P_c + g_i P_p), \quad i=1,...,K,
\end{equation}
where $0<\eta_i \leq 1$ is the energy harvesting efficiency at CU$_i$, for which we assume $\eta_i = 1$, $i=1,...,K$, in the sequel for convenience, unless otherwise noted; $P_p$ and $P_c$ denote the transmit power of the PT and the H-AP, respectively. In the remaining WIT phase with $(1-\tau)$ amount of time, which is assigned to the CUs' UL WIT to the H-AP, the achievable sum-throughput (in bits/sec/Hz) at the H-AP based on optimal successive decoding is given by \cite{B_CT:1991} 
\begin{equation} \label{Eq:CR Rate}
(1-\tau) \log_2\l(1 + \sum_{i=1}^K\frac{h_i}{\sigma^2_c + \hat{g} P_p} \frac{e_i}{1-\tau}\r),
\end{equation}
where $\sigma_c^2$ denotes the receiver noise at the H-AP, and $0\leq e_i\leq \tau Q_i$, $i=1,...,K$, indicates the amount of harvested energy used for CU$_i$'s WIT, with $\tau Q_i$ denoting the total harvested energy at each CU$_i$ during the WET phase.\footnote{In this paper, we assume that the main energy consumption at the transmitter in each CU is due to its WIT, and thus the energy consumptions due to other sources such as circuit energy are ignored for simplicity. Furthermore, we focus on the transmission optimization in a single block for a given set of channel realization without considering initial energy at each CU. Nonetheless, the solutions we obtained for (P1) and (P2) can easily be extended to solve more general problems with initial energy consideration at each CU. } It should be noted that \eqref{Eq:CR Rate} takes into account the interference with power $\hat{g}P_p$, experienced at the H-AP due to the PT's WIT.\footnote{It is worth noting that if the primary user's signal can be opportunistically decoded and subtracted \cite{J_ZC:2012} at the H-AP, the secondary sum-throughput can be further improved as compared to that given in \eqref{Eq:CR Rate}. However, we do not consider this scheme further in this paper for the simplicity of the U-CWPCN design without assuming any knowledge of the primary transmission. } 

In U-CWPCN, the H-AP's DL WET and the CUs' UL WIT are regulated by the given ITC \cite{J_ZLC:2010}, such that the interference power at the PR is kept no larger than a predefined threshold, denoted by $\Gamma\geq 0$. Accordingly, during the WET phase, we set the transmit power of the H-AP as 
\begin{equation} \label{Eq:P_c Underlay}
P_c = \min\l(\frac{\Gamma}{\hat{h}}, \Pmax\r), 
\end{equation}
where $\Pmax$ is the maximum transmit power of the H-AP, while during the WIT phase, the ITC is applied so that
\begin{equation}\label{Eq:ITConstraint}
\sum_{i=1}^K \hat{h}_i \frac{e_i}{1-\tau} \leq \Gamma.
\end{equation}
Note that to apply the ITC, we assume that the channel power gains from the H-AP to PR (i.e., $\hat{h}$) and from CU$_i$ to PR (i.e., $\hat{h}_i$, $i=1,...,K$), are perfectly known at the CWPCN.

For U-CWPCN, the sum-throughput maximization problem that jointly optimizes the DL-UL time allocation $\tau$ and each CU$_i$'s transmit energy $e_i$ can be formulated as follows.
\begin{align}
\mathrm{(P1)}:~\mathop{\mathtt{Maximize}}_{\tau,\{e_i\}} &~~  (1-\tau) \log_2\l(1+\frac{1}{1-\tau}\sum_{i=1}^K \gamma_i e_i\r)  \label{Eq:P1 Objective}\\
\mathtt{subject\; to}
&~~0\leq \tau \leq 1 \nn\\
&~~0\leq e_i \leq \tau Q_i, \; i=1,...,K  \label{Eq:P1Constraint2}\\
&~~\sum_{i=1}^K \hat{h}_i e_i \leq (1-\tau)\Gamma,   \label{Eq:P1Constraint3} 
\end{align}
where in  \eqref{Eq:P1 Objective} we have  $\gamma_i \triangleq \frac{h_i}{\sigma^2_c + \hat{g} P_p}$ for notational convenience. Note that unlike the sum-throughput maximization in a stand-alone WPCN studied in \cite{J_JZ:2014}, in (P1) we need to consider the effect of the new ITC in \eqref{Eq:P1Constraint3} for the optimal time and energy allocation in the U-CWPCN. Besides, in \cite{J_JZ:2014} the suboptimal TDMA based UL WIT is assumed, while in (P1) we study the capacity-achieving sum-throughput. It is also worth pointing out that, as compared to the sum-capacity maximization problem studied in \cite{J_ZCL:2009} for conventional CR-based multiple-access channel where the secondary users' UL power allocation is optimized under the ITC, in (P1) the DL-UL time allocation needs to be jointly optimized with each CU$_i$'s UL transmit power due to the new consideration of the WET in the U-CWPCN.

Note that problem (P1) is a convex optimization problem, since the objective function \eqref{Eq:P1 Objective} can be shown to be concave, and the constraints are all linear. In Section~\ref{Section:Underlay}, we present an efficient optimal solution to (P1) in detail.

\subsection{Overlay Based CWPCN (O-CWPCN)} \label{Subsection:Overlay}
Next, consider the O-CWPCN. In overlay based CR, it is assumed that the coexisting secondary users have noncausal knowledge of the primary users' codewords to be transmitted; as a result, the secondary users can exploit such knowledge to assist primary  transmission by cooperatively relaying the primary message, as well as to cancel the interference at the secondary receivers due to the primary transmission  (for more details, see, e.g., \cite{J_GJMS:2009, J_DMT:2006, J_JV:2009}). It is worth noting that the primary messages can be made known to the secondary users through dedicated links connecting the primary and secondary transmitters. This is practically feasible in e.g., the coordinated multipoint (CoMP) transmission where multiple base stations (BSs) in a cellular network share their transmit messages and channel state information (CSI) via backhaul links to enable cooperative transmissions in order to mitigate/utilize the inter-cell interference \cite{J_GHHSSY:2010}. Similarly, in an overlay based CR, secondary users are considered as cooperative transmission nodes, and thus primary users are willing to share their transmit messages/CSI with the secondary users to utilize their signals in a beneficial way to achieve cooperative transmission.

Based on this scheme, our proposed O-CWPCN is explained as follows. Similar to the U-CWPCN, it is assumed that the first $\tau$ amount of time in each block transmission is for WET, and the remaining $(1-\tau)$ time is for  WIT, as shown in Fig~\ref{Fig:Overlay}. However, unlike U-CWPCN, it is assumed in O-CWPCN that the H-AP has the knowledge of the PT's message to be transmitted as well as the primary link's channel at the beginning of each block. With such knowledge, the H-AP can cooperatively send the PT's message using its WET signal in the WET phase to make it coherently added with the PT's signal at the PR, in return for higher transmit power for its WET (as compared to $P_c$ given in \eqref{Eq:P_c Underlay} in the U-CWPCN case)  as well as CUs' WIT. Accordingly, the average achievable rate of the primary link in a given block can be expressed as
\begin{align} \label{Eq:PrimaryRate}
R_{\text{P-O}}(\tau,\{e_i\}) &= \tau r_1   \nn\\ 
& + (1-\tau)\log_2\l(1+\frac{g P_p}{\sigma_p^2 +\frac{1}{1-\tau}\sum_{i=1}^K\hat{h}_i e_i}\r),
\end{align}
where $r_1\triangleq \log_2\l(1+\frac{g P_p + \hat{h}P_c + 2\sqrt{g\hat{h}P_p P_c}}{\sigma_p^2}\r)$ with $\sigma^2_p$ denoting the receiver noise power at the PR. In \eqref{Eq:PrimaryRate}, $\tau r_1$ and $(1-\tau) \log_2\l(1+\frac{g P_p}{\sigma_p^2 +\frac{1}{1-\tau}\sum_{i=1}^K\hat{h}_i e_i}\r)$ denote the primary user rates during the WET phase and WIT phase, respectively. It can be seen from \eqref{Eq:PrimaryRate} that, since $r_1 \geq \log_2\l(1+\frac{g P_p}{\sigma_p^2}\r)$, where $ \log_2\l(1+\frac{g P_p}{\sigma_p^2}\r)$ is the maximum primary user rate without the presence of the secondary CWPCN, the primary link may be able to achieve $R_{\text{P-O}}(\tau,\{e_i\}) \geq \log_2\l(1+\frac{g P_p}{\sigma_p^2}\r)$ with a proper choice of $(\tau,\{e_i\})$, i.e., it is possible that the primary rate performance can even be improved in O-CWPCN due to the H-AP's cooperative transmission in the WET phase. Also note that in this case the transmit power of the H-AP for its DL WET can be set to the maximum value, i.e., $P_c = \Pmax$, unlike that given in \eqref{Eq:P_c Underlay} for the case of U-CWPCN.

With \eqref{Eq:PrimaryRate}, we can use the PRC (instead of the ITC given in  \eqref{Eq:P1Constraint3} in U-CWPCN), such that the primary can achieve an average rate $R_{\text{P-O}}(\tau,\{e_i\})$ that is no less than a predefined threshold, denoted by $\bar{R}\geq 0$. Accordingly, we formulate the following sum-throughput maximization problem for O-CWPCN.
\begin{align}
\mathrm{(P2)}:~\mathop{\mathtt{Maximize}}_{\tau,\{e_i\}} &~~  (1-\tau) \log_2\l(1+\frac{1}{1-\tau}\sum_{i=1}^K \hat{\gamma}_i e_i\r) \label{Eq:P2Objective} \\
\mathtt{subject\; to}
&~~0\leq \tau \leq 1 \nn\\
&~~0\leq e_i \leq \tau Q_i, \; i=1,...,K  \label{Eq:P2Constraint2}\\
&~~  R_{\text{P-O}}(\tau,\{e_i\}) \geq \bar{R}, \label{Eq:P2Constraint3} 
\end{align}
where in \eqref{Eq:P2Objective} we have $\hat{\gamma_i} \triangleq  h_i/\sigma_c^2$, since the interference signal due to the PT's WIT can be perfectly canceled at the H-AP's receiver in the WIT phase, thanks to the knowledge of the PT's message and the channel knowledge from the PT to the H-AP.

 Unlike (P1), problem (P2) is a non-convex optimization problem, due to the PRC in \eqref{Eq:P2Constraint3}. Nevertheless, we will provide an efficient way to obtain the optimal solution of (P2) in Section~\ref{Section:Overlay}.

\subsection{Performance and Complexity Comparison}
In this subsection, we analyze the performance of the proposed U-CWPCN and O-CWPCN models, and compare their complexity for implementation in terms of the channel state information (CSI) required. For the purpose of exposition, we define the feasible solution sets of (P1) and (P2) as
\begin{align*}
&\cF_U(\Gamma) = \bigg\{(\tau,\{e_i\}): 0\leq \tau \leq 1, 0\leq e_i \leq \tau Q_i, \; i=1,...,K, \nn\\ 
&\qquad \sum_{i=1}^K \hat{h}_i e_i \leq (1-\tau)\Gamma, P_c = \min\l(\frac{\Gamma}{\hat{h}}, \Pmax\r) \bigg\}, 
\end{align*}
\begin{align*}
\cF_O(\bar{R}) = &\{(\tau,\{e_i\}): 0\leq \tau \leq 1, 0\leq e_i \leq \tau Q_i, \; i=1,...,K, \nn\\ & R_{\text{P-O}}(\tau,\{e_i\}) \geq \bar{R},P_c = \Pmax \}, 
\end{align*}
respectively. In addition, for notational brevity, we define a function
\begin{equation} \label{Eq:RateCWPCN}
R_{\text{C}}(\tau,\{e_i\},\{\delta_i\}) = (1-\tau)\log_2\l(1+\frac{1}{1-\tau}\sum_{i=1}^K \delta_i e_i\r),
\end{equation}
to represent the CWPCN's sum-throughput with different values of $\{\delta_i\}$ in each spectrum-sharing case. Then, our performance metric is the primary and secondary systems' \emph{achievable rate region}, defined as the set of rate-pairs for the primary link's rate and the CWPCN's sum-throughput that can be achievable at the same time under the associated constraint. Such achievable rate regions for U-CWPCN and O-CWPCN can be  defined as (see (P1) and (P2))
\begin{align} \label{Eq:RateRegionU} 
\cR_U = &\bigcup_{\Gamma\geq 0}\{(R_p, R_c): R_p \leq R_\text{P-U}(\tau,\{e_i\}), \nn\\ 
&\quad R_c\leq R_{\text{C}}(\tau,\{e_i\}, \{\gamma_i\}), (\tau,\{e_i\})\in \cF_U(\Gamma)\}, 
\end{align}
\begin{align} \label{Eq:RateRegionO}
\cR_O = &\bigcup_{\bar{R}\geq 0}\{(R_p, R_c): R_p \leq R_{\text{P-O}}(\tau,\{e_i\}), \nn\\ 
& \quad R_c\leq R_{\text{C}}(\tau,\{e_i\}, \{\hat{\gamma}_i\}), (\tau,\{e_i\})\in \cF_O(\bar{R})\},  
\end{align}
respectively, where in \eqref{Eq:RateRegionU} we have defined  
\begin{align*}
R_\text{P-U} & (\tau,\{e_i\})  =  \tau \log_2\l(1+\frac{g P_p}{\sigma_p^2 + \hat{h}P_c }\r)  \nn\\ 
& + (1-\tau)\log_2\l(1+\frac{g P_p}{\sigma_p^2 +\frac{1}{1-\tau}\sum_{i=1}^K\hat{h}_i e_i}\r),
\end{align*}
which represents the achievable rate of the primary link in the U-CWPCN, for a given $(\tau,\{e_i\})$. 

\begin{table}
\caption{Required CSI at Different  CWPCNs  (O: Required. X: Not required)}
\centering
\begin{tabular}{|l|cc| |cc || cc |}
\hline
        & \multicolumn{2}{c}{U-CWPCN} & \multicolumn{2}{c|}{O-CWPCN} \\
\hline
Channels & Power Gain      & Phase   & Power Gain      & Phase  \\
\hline 
$\hat{h}$: From H-AP to PR     		  & O               & X          & O               & O         \\
$h_i$: From H-AP to CU$_i$    		  & O               & O         & O               & O         \\
$\hat{h}_i$:  From CU$_i$ to PR        & O               & X          & O               & X         \\
$g_i$: From PT to CU$_i$   	         & O               & X          & O               & X         \\
$g$: From PT to PR      		         & X               & X          & O               & O         \\
$\hat{g}$: From PT to H-AP               & O               & X          & O               & O       \\ \hline
\end{tabular} 
\end{table}

\begin{proposition} \label{Proposition:Performance}
It holds that $\cR_U\subseteq \cR_O$, i.e., the O-CWPCN outperforms the U-CWPCN in general.
\end{proposition}
\begin{proof}
See Appendix~\ref{Proof:Proposition:Performance}.
\end{proof}

Proposition~\ref{Proposition:Performance} indicates that, for instance, with the same sum-throughput of the CWPCN, the maximum achievable rate of the primary user is in general larger in O-CWPCN than that in U-CWPCN. This is due to the fully cooperative WET/WIT design with the primary user's transmission in the O-CWPCN. However, it should be noted that the O-CWPCN requires additional CSI as well as primary user message knowledge for implementation. Table II thus compares the required CSI at each  of  the two CWPCNs in practice. It is worth pointing out that the O-CWPCN needs to know additionally the phase of the channels of $\hat{h}$ and $g$ for implementing the cooperative transmission at the H-AP's transmitter with the PT's WIT (see \eqref{Eq:PrimaryRate}), as well as that of $\hat{g}$ for the interference cancellation at the H-AP's receiver. 

In this paper, we assume that the perfect channel knowledge is available at the CWPCN as indicated in Table II to study the performance limit of the proposed CWPCN schemes. In practice, for both the U-CWPCN and O-CWPCN, the CSI on the channels within the CWPCN (i.e., the CSI of $h_i$'s) can be obtained by classic channel training, estimation and feedback mechanisms. For the U-CWPCN, the channel power gains between the PT and the CWPCN (i.e., $g_i$'s and $\hat{g}$) can be obtained by the CWPCN via estimating the received signal power from the PT provided that it knows \emph{a priori} the transmit power of the PT, while those between the CWPCN and the PR (i.e., $\hat{h}$ and $\hat{h}_i$'s) can be obtained by the CWPCN by similarly estimating the received signal power when the PR transmits and assuming the channel reciprocity. On the other hand, for the O-CWPCN, the additional CSI on the primary link (i.e., the CSI of $g$) can be obtained by the CWPCN via a dedicated feedback link from the PT, as explained in Section~\ref{Subsection:Overlay}.

\section{Sum-Throughput Maximization in Underlay Based CWPCN} \label{Section:Underlay}
In this section, we present the optimal solution to problem (P1), i.e., the sum-throughput maximization problem for the case of U-CWPCN. By fixing $\tau = \bar{\tau}$, $0<\bar{\tau}<1$,  (P1) reduces to the following optimization problem.
\begin{align}
\mathrm{(P1.1)}:~\mathop{\mathtt{Maximize}}_{\{e_i\}} &~~(1-\bar{\tau})\log_2\l( 1+ \frac{1}{1-\bar{\tau}} \sum_{i=1}^K \gamma_i e_i  \r) \nn\\
\mathtt{subject\; to}
&~~0\leq e_i \leq \bar{\tau} Q_i, \; i=1,...,K  \nn\\
&~~\sum_{i=1}^K \hat{h}_i e_i \leq (1-\bar{\tau})\Gamma. \label{P1.1Constraint2}
\end{align}
It can be seen that (P1.1) is also a convex optimization problem, similar to (P1). Hence, it must satisfy the Karush-Kuhn-Tacker (KKT) conditions \cite{B_BV:2004}, from which we can derive the optimal solution of (P1.1), denoted by $\{\bar{e}_i\}$, which is given in the following proposition.
\begin{proposition} \label{Theorem:Underlay}
Let the set of indices $\{(1),(2),...,(K)\}$ denote $\l\{\frac{\gamma_i}{\hat{h}_i}\r\}$, $i=1,...,K$, in decreasing order, i.e., $\frac{\gamma_{(1)}}{\hat{h}_{(1)}} > \frac{\gamma_{(2)}}{\hat{h}_{(2)}} > ... > \frac{\gamma_{(K)}}{\hat{h}_{(K)}}$. Then, the optimal solution to (P1.1) is given as follows.
\begin{enumerate}
\item We have  $\bar{e}_i = \bar{\tau} Q_i$, $i=1,...,K$, i.e., all the CUs transmit with full power, if and only if
\begin{equation} \label{Eq:tau_K+1}
0 <  \bar{\tau} \leq \frac{\Gamma}{\Gamma + \sum_{i=1}^K \hat{h}_i Q_i}.
\end{equation}

\item Otherwise, we have
\begin{align} \label{Eq:e_i_optimal}
&\bar{e}_{(i)} = \l\{\begin{aligned}
&\bar{\tau}Q_{(i)}, \quad \mbox{if } i=1,...,k-1, \\
&\frac{1}{\hat{h}_{(i)}}\l(\Gamma(1-\bar{\tau}) - \bar{\tau}\sum_{i=1}^{k-1}\hat{h}_{(i)}Q_{(i)}\r), \; \mbox{if } i=k, \\
&0, \quad \mbox{if } i=k+1,...,K,
\end{aligned}
\r.
\end{align}
for a given $k=1,...,K$, i.e., at most one CU transmits with fractional power, and the others either transmit with full power or do not transmit, where $k$ satisfies
\begin{equation} \label{Eq:tau_k}
\frac{\Gamma}{\Gamma + \sum_{i=1}^{k} \hat{h}_{(i)} Q_{(i)}} < \bar{\tau} \leq \frac{\Gamma}{\Gamma + \sum_{i=1}^{k-1} \hat{h}_{(i)} Q_{(i)}}.
\end{equation}
\end{enumerate}
\end{proposition}
\begin{proof}
See Appendix~\ref{Proof:Underlay}.
\end{proof}

\begin{figure} 
\centering
\includegraphics[width=8.5cm]{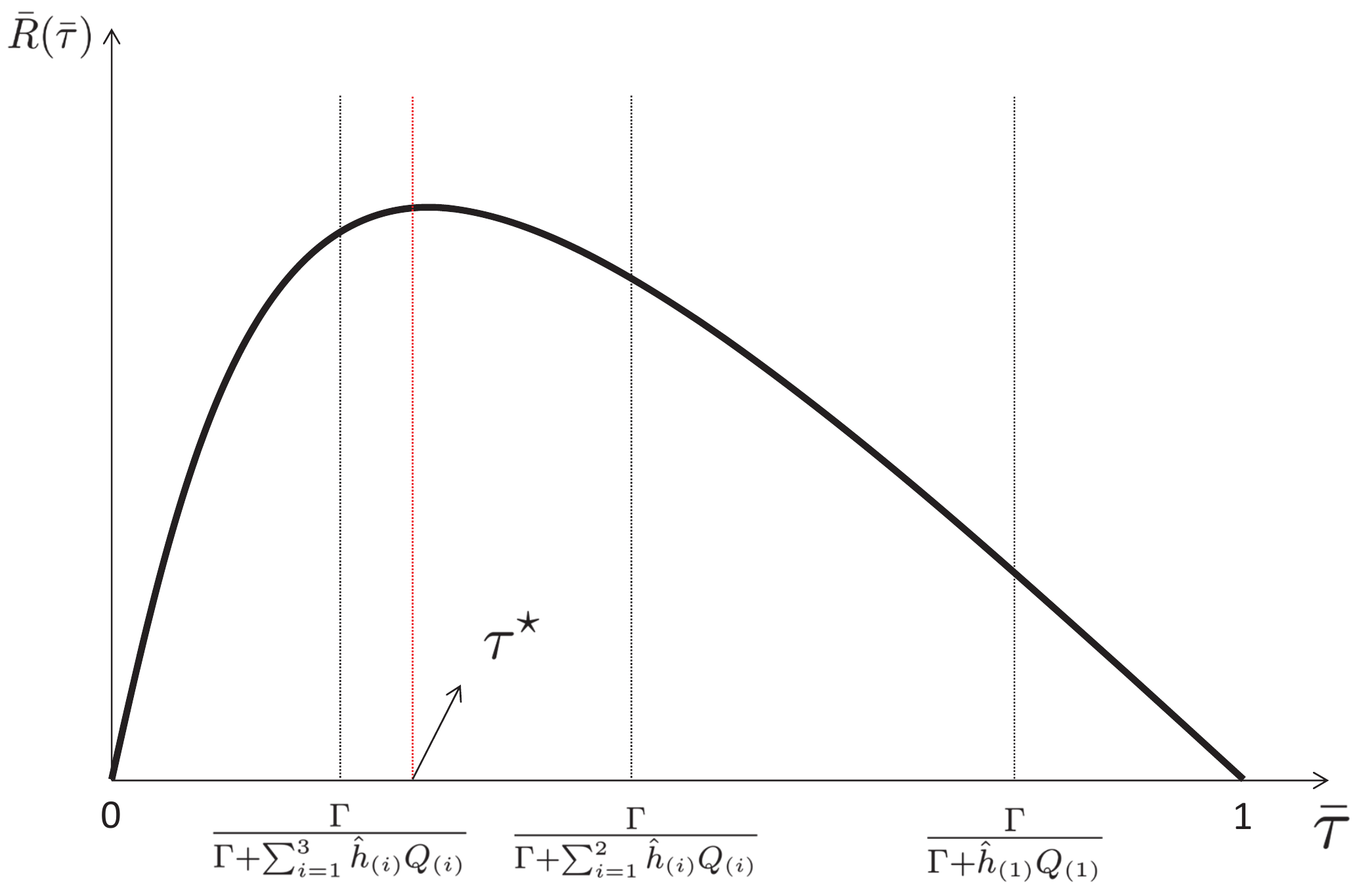}
\caption{Illustration of $\bar{R}(\bar{\tau})$. } \label{Fig:Tau}
\end{figure}

Thus, we can solve (P1.1) for any given $0<\bar{\tau}<1$. Now, let $\bar{R}(\bar{\tau})$ denote the optimal value of (P1.1) with a given $0<\bar{\tau}<1$. By plugging $\bar{e}_i$'s obtained in Proposition~\ref{Theorem:Underlay} into the objective function of (P1.1), $\bar{R}(\bar{\tau})$ can be expressed as 
\begin{align} \label{Eq:f_k}
\bar{R}(\bar{\tau}) =  \l\{\begin{aligned}
& (1-\bar{\tau}) \log_2\l(1+\frac{\bar{\tau}}{1-\bar{\tau}}\sum_{i=1}^K \gamma_i Q_i\r), \\ 
& \mbox{if } 0 <  \bar{\tau} \leq \frac{\Gamma}{\Gamma + \sum_{i=1}^K \hat{h}_i Q_i},  \\
& (1-\bar{\tau}) \log_2\bigg(1+\frac{\bar{\tau}}{1-\bar{\tau}}\sum_{i=1}^{k-1}\gamma_{(i)}Q_{(i)}  \\ 
& \quad + \frac{1}{1-\bar{\tau}}\frac{\gamma_{(k)}}{\hat{h}_{(k)}}\l[\Gamma(1-\bar{\tau}) - \bar{\tau}\sum_{i=1}^{k-1}\hat{h}_{(i)}Q_{(i)}\r]\bigg),\\ 
&\mbox{if } \frac{\Gamma}{\Gamma + \sum_{i=1}^{k} \hat{h}_{(i)} Q_{(i)}} < \bar{\tau} \leq \frac{\Gamma}{\Gamma + \sum_{i=1}^{k-1} \hat{h}_{(i)} Q_{(i)}}, \\&\qquad\qquad\qquad\qquad\qquad\qquad\qquad k=1,...,K.
\end{aligned}
\r.
\end{align}
It then follows that we can solve (P1) optimally by maximizing $\bar{R}(\bar{\tau})$ over $0<\bar{\tau}< 1$, i.e.,
\begin{equation} \label{Eq:P1_tau}
\tau^\star = \arg\max_{0<\bar{\tau}<1} \bar{R}(\bar{\tau}),
\end{equation}
where $\tau^\star$ denotes the optimal $\tau$ of (P1). It can be shown that $\bar{R}(\bar{\tau})$ is a concave function over $0<\bar{\tau}<1$, as illustrated in Fig.~\ref{Fig:Tau} for the case of $K=3$, and thus \eqref{Eq:P1_tau} can be efficiently solved by e.g., Newton's method. With the obtained $\tau^\star$, the resulting optimal $\{e_i\}$ of (P1), denoted by $\{e_i^\star\}$, can be directly obtained from Proposition~\ref{Theorem:Underlay} with $\bar{\tau}=\tau^\star$. Problem (P1) is thus solved.

\section{Sum-Throughput Maximization in Overlay Based CWPCN} \label{Section:Overlay}
In this section, we solve problem (P2), i.e., the sum-throughput maximization for the case of O-CWPCN. As mentioned in Section~\ref{Subsection:Overlay}, problem (P2) is non-convex due to the PRC in \eqref{Eq:P2Constraint3}. Nonetheless, in the following we show that the problem can be solved optimally by using the technique so-called \emph{active interference-temperature control} \cite{J_ZLC:2010}. The idea is mainly from the observation that, at the optimal solution of (P2), there must exist an unique interference-temperature power level at the PR. As will be shown later in this section, (P2) can thus be solved via solving a sequence of time/power allocation problems in the O-CWPCN each subject to a given ITC and iteratively searching for the optimal ITC value to maximize the sum-throughput, where each subproblem can be solved with a similar solution to that of (P1) in the case of U-CWPCN in Section~\ref{Section:Underlay}.

The details are given as follows. For the purpose of exposition, we consider the following problem (P2-a) which has the same objective function with (P2).
\begin{align}
\text{(P2-a)}:\nn\\
~\mathop{\mathtt{Maximize}}_{\tau,\{e_i\}} &~~  (1-\tau) \log_2\l(1+\frac{1}{1-\tau}\sum_{i=1}^K \hat{\gamma}_i e_i\r) \nn  \\
\mathtt{subject\; to}
&~~0\leq \tau \leq 1 \label{Eq:P2.1Constraint1} \\
&~~0\leq e_i \leq \tau Q_i, \; i=1,...,K  \label{Eq:P2.1Constraint2}\\
&~~  \tau r_1 + (1-\tau) \log_2\l(1+\frac{gP_p}{\sigma_p^2 + \Gamma_0}\r) \geq \bar{R} \label{Eq:P2.1Constraint3} \\
&~~ \frac{1}{1-\tau}\sum_{i=1}^K \hat{h}_i e_i \leq \Gamma_0, \label{Eq:P2.1Constraint4}
\end{align}
where $\Gamma_0 \geq 0$. Let $R_\text{a}^\star(\Gamma_0)$ denote the optimal value of (P2-a) with a given $\Gamma_0\geq 0$. Furthermore, define
\begin{equation*}
\tilde{\Gamma}_0 = \arg\max_{\Gamma_0\geq 0} R_\text{a}^\star(\Gamma_0),
\end{equation*}
i.e., $R_\text{a}^\star(\tilde{\Gamma}_0)\geq R_\text{a}^\star(\Gamma_0)$, $\forall \Gamma_0\geq 0$. The following proposition shows that the optimal values of (P2-a) and (P2) become same when $\Gamma_0 = \tilde{\Gamma}_0$.

\begin{proposition} \label{Proposition:ActiveInterference}
Let  $R^\star$ denote the optimal value of (P2). It then follows that $R^\star = R_\text{a}^\star(\tilde{\Gamma}_0)$. 
\end{proposition}
\begin{proof}
See Appendix~\ref{Proof:Proposition:ActiveInterference}.
\end{proof}

In other words, (P2-a) becomes equivalent to (P2) when $\Gamma_0 = \tilde{\Gamma}_0$. Thus, (P2) can be solved by first solving (P2-a) with any given $\Gamma_0 \geq 0$, and then with obtained $R_\text{a}^\star(\Gamma_0)$, solving
\begin{equation} \label{Eq:OverlayGamma}
R^\star = \max_{\Gamma_0 \geq 0} R_\text{a}^\star(\Gamma_0).
\end{equation}
Thus, in the rest of this section, we present the solution to (P2-a) as well as how to efficiently search for $\Gamma_0 \geq 0$ that maximizes $R_\text{a}^\star(\Gamma_0)$ to solve \eqref{Eq:OverlayGamma}.

First, to solve (P2-a), we show the connection between (P2-a) and (P1) which has been solved in Section~\ref{Section:Underlay}. Note that the constraint \eqref{Eq:P2.1Constraint3} in (P2-a) can be simplified as $\tau \geq \frac{\bar{R}-r_2(\Gamma_0)}{r_1-r_2(\Gamma_0)}$, where $r_2(\Gamma_0) \triangleq \log_2\l(1+\frac{gP_p}{\sigma_p^2 + \Gamma_0}\r)$ for notational convenience. As a result, the constraints \eqref{Eq:P2.1Constraint1} and \eqref{Eq:P2.1Constraint3} can be combined as 
\begin{equation*}
\l[\frac{\bar{R}-r_2(\Gamma_0)}{r_1-r_2(\Gamma_0)}\r]^+ \leq \tau \leq 1,
\end{equation*}
where $[a]^+ \triangleq \max(0,a)$. Since $\l[\frac{\bar{R}-r_2(\Gamma_0)}{r_1-r_2(\Gamma_0)}\r]^+\geq 0$, it follows that (P2-a) becomes a special case of (P1) with $\gamma_i$ and $\Gamma $ in (P1) replaced by $\hat{\gamma}_i$ and $\Gamma = \Gamma_0$, respectively. Therefore, the optimal solution of (P2-a) with a given $\Gamma_0\geq 0$, denoted by  $(\tau(\Gamma_0), \{e_i(\Gamma_0)\})$, can be obtained by slightly modifying the solution of (P1), explained as follows. Let $\hat{\tau}^\star$ denote the optimal solution of (P1) with $\gamma_i$ replaced by $\hat{\gamma}_i$ and $\Gamma = \Gamma_0$. We  have seen in Section~\ref{Section:Underlay} that $\bar{R}(\bar{\tau})$ given in \eqref{Eq:f_k} is concave over $0\leq \bar{\tau}\leq 1$ and thus the optimal $\bar{\tau}$ can be efficiently found by e.g., Newton's method (see Fig.~\ref{Fig:Tau}), based on which $\hat{\tau}^\star$ can be found similarly. Thus, it can be easily verified that
\begin{equation} \label{Eq:tau_Gamma}
\tau(\Gamma_0) = \l\{\begin{aligned}
&\hat{\tau}^\star,  \quad\quad\quad\quad\quad\text{if }\; \frac{\bar{R}-r_2(\Gamma_0)}{r_1-r_2(\Gamma_0)} < \hat{\tau}^\star \\
& \frac{\bar{R}-r_2(\Gamma_0)}{r_1 - r_2(\Gamma_0)},  \quad\text{otherwise},
\end{aligned}
\r.
\end{equation}
\begin{align*}
&e_{(i)}(\Gamma_0) = \l\{\begin{aligned}
&\tau(\Gamma_0) Q_{(i)}, \quad \mbox{if } i=1,...,k-1, \\
&\frac{1}{\hat{h}_{(i)}}\l(\Gamma_0(1-\tau(\Gamma_0)) - \tau(\Gamma_0)\sum_{i=1}^{k-1}\hat{h}_{(i)}Q_{(i)}\r), \nn\\ &\qquad\qquad\qquad\mbox{if } i=k, \\
&0, \quad \mbox{if } i=k+1,...,K. \\
\end{aligned}
\r.
\end{align*}
Problem (P2-a) is thus solved. 

Next, we discuss how to reduce the feasible set of $\Gamma_0$ to search for solving \eqref{Eq:OverlayGamma}, based on the solution of (P2-a). We have seen in Proposition~\ref{Proposition:ActiveInterference} that with $\Gamma_0 = \tilde{\Gamma}_0$, (P2-a) becomes equivalent to (P2), and it can be inferred that with the corresponding optimal solution of (P2-a), i.e., $(\tau(\tilde{\Gamma}_0), \{e_i(\tilde{\Gamma}_0)\})$, the interference constraint \eqref{Eq:P2.1Constraint4} should hold with equality, i.e., 
\begin{equation*}
\frac{1}{1-\tau(\tilde{\Gamma}_0)}\sum_{i=1}^K \hat{h}_i e_i(\tilde{\Gamma}_0) = \tilde{\Gamma}_0.
\end{equation*}
From the above, we observe that if we are only interested in the values of $\Gamma_0$ that could be the optimal solution of (P2), it is sufficient to search over the subset of $\{\Gamma_0:\Gamma_0\geq 0\}$ that makes the interference constraint \eqref{Eq:P2.1Constraint4}  tight. The following proposition then gives a set of $\Gamma_0$ that is of interest.
\begin{proposition} \label{Proposition:Overlay}
With the optimal solution of (P2-a), the constraint \eqref{Eq:P2.1Constraint4} is tight if and only if
\begin{align} \label{Eq:GammaTightSet}
&\Gamma_0 \in \l\{\Gamma_0 \geq 0: \tilde{\tau}_K\geq \frac{\Gamma_0}{\Gamma_0 + \sum_{i=1}^K \hat{h}_i Q_i}\r\}  \nn\\ 
&\cup\l\{\Gamma_0 \geq 0: \frac{\bar{R}-r_2(\Gamma_0)}{r_1 - r_2(\Gamma_0)}\geq \frac{\Gamma_0}{\Gamma_0 + \sum_{i=1}^K \hat{h}_i Q_i} \r\} \triangleq \cF_{\Gamma_0},
\end{align}
where $\tilde{\tau}_K$ is defined as the solution of the following equation 
\begin{equation*}
\frac{\partial}{\partial \tau} \l[(1-\tau) \log_2\l(1+\frac{\tau}{1-\tau}\sum_{i=1}^K  \hat{\gamma}_i Q_i\r) \r]= 0.
\end{equation*}
\end{proposition}
\begin{proof}
See Appendix~\ref{Proof:Proposition:Overlay}.
\end{proof}

Proposition~\ref{Proposition:Overlay} states that $\Gamma_0\in \cF_{\Gamma_0}$ if and only if $\frac{1}{1-\tau(\Gamma_0)}\sum_{i=1}^K \hat{h}_i e_i(\Gamma_0) = \Gamma_0$.

\begin{figure} 
\centering
\includegraphics[width=8.5cm]{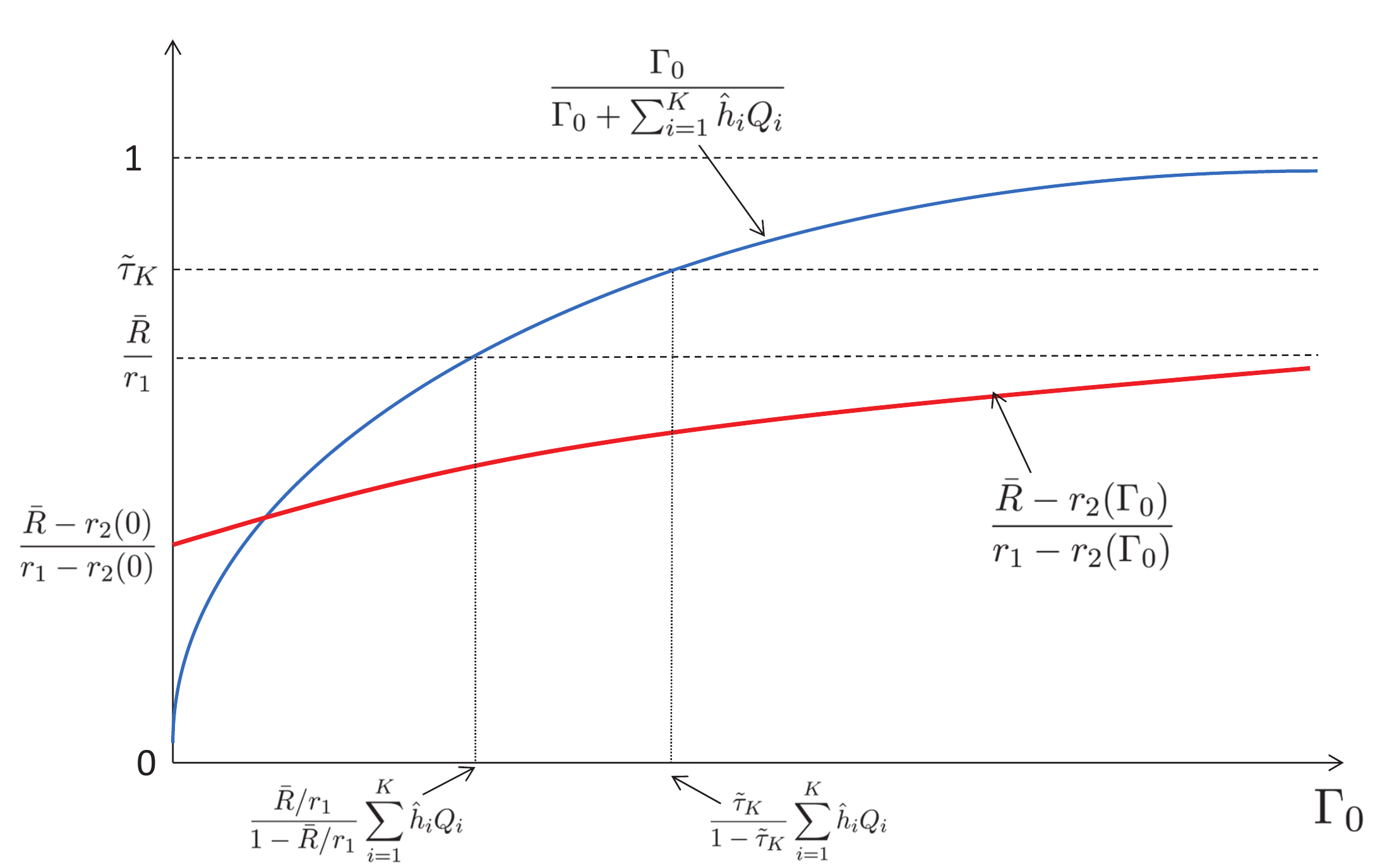}
\caption{Illustration of $\frac{\bar{R} - r_2(\Gamma_0)}{r_1 - r_2(\Gamma_0)}$ and $\frac{\Gamma_0}{\Gamma_0 + \sum_{i=1}^K \hat{h}_i Q_i}$ with respect to $\Gamma_0$. } \label{Fig:Gamma}
\end{figure}

Furthermore, we show that there is an upper bound of $\Gamma_0$ below which $\cF_{\Gamma_0}$ is included. First, it is noted that both $\frac{\bar{R}-r_2(\Gamma_0)}{r_1 - r_2(\Gamma_0)}$ and $\frac{\Gamma_0}{\Gamma_0 + \sum_{i=1}^K \hat{h}_i Q_i}$ are monotonically increasing functions of $\Gamma_0 \geq 0$. Furthermore, as $\Gamma_0 \rightarrow \infty$, it follows that  $\frac{\bar{R} - r_2(\Gamma_0)}{r_1 - r_2(\Gamma_0)} \rightarrow \frac{\bar{R}}{r_1}$ and  $\frac{\Gamma_0}{\Gamma_0 + \sum_{i=1}^K \hat{h}_i Q_i} \rightarrow 1$. Note that $\frac{\bar{R}}{r_1}\leq 1$ must hold, since \eqref{Eq:P2.1Constraint4} cannot be satisfied otherwise. As a result, the inequality $\frac{\bar{R}-r_2(\Gamma_0)}{r_1 - r_2(\Gamma_0)} < \frac{\Gamma_0}{\Gamma_0 + \sum_{i=1}^K \hat{h}_i Q_i}$ always holds when $\frac{\Gamma_0}{\Gamma_0 + \sum_{i=1}^K \hat{h}_i Q_i} \geq \frac{\bar{R}}{r_1}$, i.e., when $\Gamma_0 \geq \frac{\bar{R}/r_1}{1-\bar{R}/r_1}\sum_{i=1}^K \hat{h}_i Q_i$ (see Fig.~\ref{Fig:Gamma} for an illustration). Consequently, it can be easily verified from Proposition~\ref{Proposition:Overlay} that, if $\Gamma_0 \in \cF_{\Gamma_0}$, then we have
\begin{align*}
\Gamma_0 &\in \l\{\Gamma_0 \geq 0: \Gamma_0\leq  \frac{\tilde{\tau}_K}{1- \tilde{\tau}_K}\sum_{i=1}^K \hat{h}_i Q_i\r\}  \nn\\
&\qquad\qquad \cup\l\{\Gamma_0 \geq 0: \Gamma_0 \leq \frac{\bar{R}/r_1}{1-\bar{R}/r_1}\sum_{i=1}^K \hat{h}_i Q_i\r\}   \\  
&=\bigg\{\Gamma_0 \geq 0: \Gamma_0 \leq \nn\\
& \qquad\max\l(\frac{ \tilde{\tau}_K}{1- \tilde{\tau}_K}\sum_{i=1}^K \hat{h}_i Q_i , \; \frac{\bar{R}/r_1}{1-\bar{R}/r_1}\sum_{i=1}^K \hat{h}_i Q_i\r)\bigg\}\\
& \triangleq \tilde{\cF}_{\Gamma_0}.
\end{align*}
In other words, $\cF_{\Gamma_0} \subseteq \tilde{\cF}_{\Gamma_0}$. Accordingly, problem \eqref{Eq:OverlayGamma} can reduce to
\begin{equation*} 
R^\star = \max_{\Gamma_0 \in \tilde{\cF}_{\Gamma_0}} R_\text{a}^\star(\Gamma_0),
\end{equation*}
where an example of $R_\text{a}^\star(\Gamma_0)$ is illustrated in Fig.~\ref{Fig:Gamma_vs_Rc} for the case of $K=5$. As can be seen from Fig.~\ref{Fig:Gamma_vs_Rc}, $R_\text{a}^\star(\Gamma_0)$ achieves its maximum when $\Gamma_0 = \tilde{\Gamma}_0$, and $\tilde{\Gamma}_0$ belongs to $\tilde{\cF}_{\Gamma_0}$, i.e., $\tilde{\Gamma}_0\in \tilde{\cF}_{\Gamma_0}$.

\begin{figure} 
\centering
\includegraphics[width=7cm]{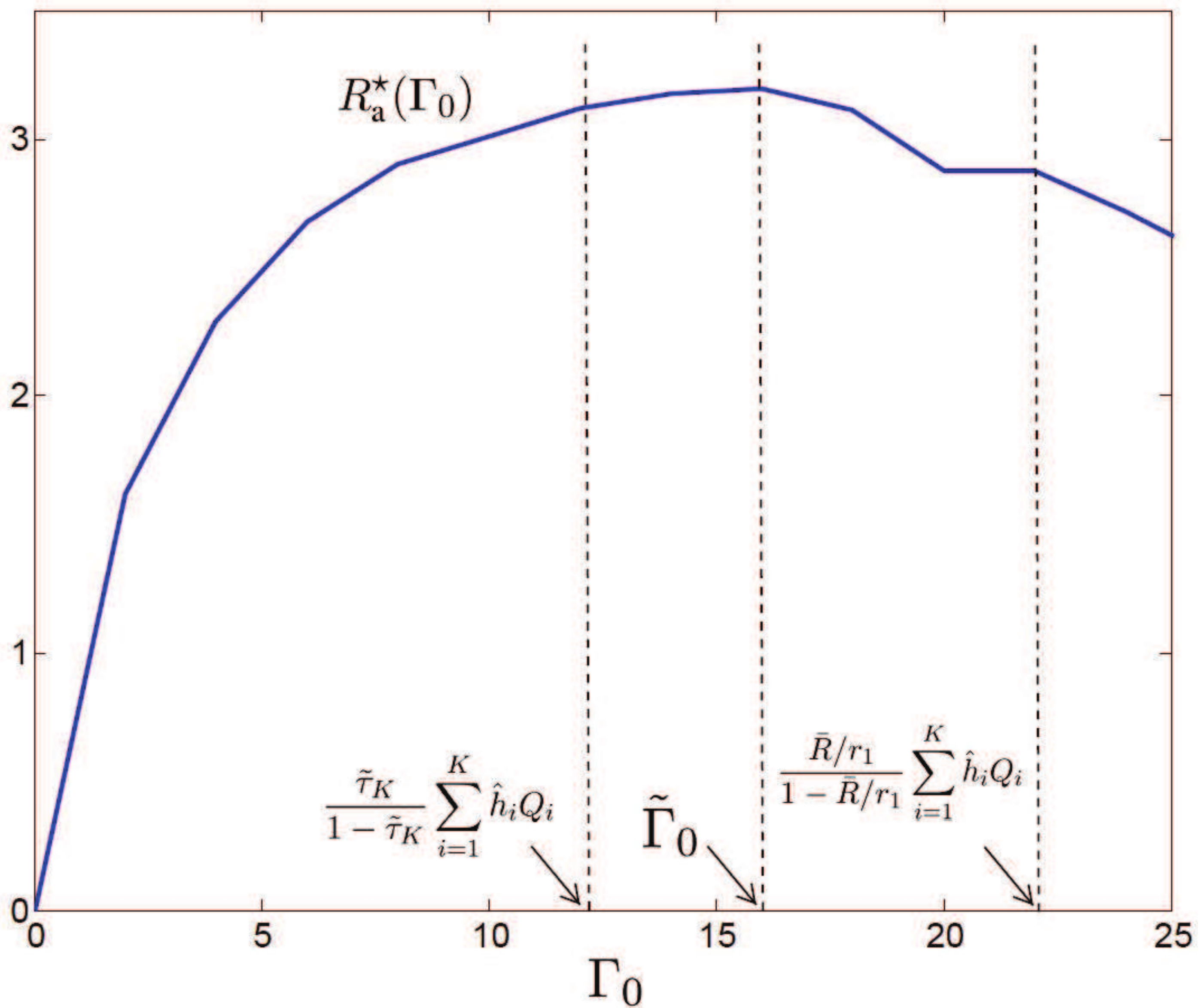}
\caption{Illustration of $R_\text{a}^\star(\Gamma_0)$, which is in general not a concave function. } \label{Fig:Gamma_vs_Rc}
\end{figure}

To summarize, (P2) can be solved optimally by iteratively solving the convex problem (P2-a) based on a similar solution of (P1), and finding $\tilde{\Gamma}_0$ via a simple one-dimensional search over $\Gamma_0 \in \tilde{\cF}_{\Gamma_0}$. Problem (P2) is thus solved. 

\section{Simulation Results} \label{Section:Simulation}

In this section, we evaluate the performance of the proposed CWPCN models by simulation. For the simulation, we set the number of CUs as $K=5$, the receiver energy harvesting efficiency as $\eta_i = 0.8$, $\forall i$, and the noise power as $\sigma^2_p = \sigma_c^2 = -90$ dBm (which corresponds to the noise power spectral density of $-160$ dBm and transmission bandwidth of $10$ MHz). As shown in Fig.~\ref{Fig:SimSetup}, we assume one-dimensional user locations for simplicity, where the PT and PR are located $200$ meters (m) apart. The H-AP and CUs are on the line segment connecting the PT and the PR, where two CUs are at the left side of the H-AP with distances of $4$m and $5$m, respectively, while the other three CUs are at the right side of the H-AP with distances of $5$m, $10$m, and $15$m, respectively. For the purpose of exposition, we compare the following three cases with different  H-AP locations: In the first case, referred to as Case 1, the H-AP is located at the middle of the PT and PR as shown in Fig.~\ref{Fig:Middle}; in the second case, referred to as Case 2, the H-AP is located closer to the PR than PT (with a distance of $170$m from the PT) as shown in Fig.~\ref{Fig:PR}; while in the third case, referred to as Case 3, the H-AP is located closer to the PT (with a distance of $30$m from the PT) as shown in Fig.~\ref{Fig:PT}.

\begin{figure}
\centering
\subfigure[Case 1]{
\centering
\includegraphics[width=8.5cm]{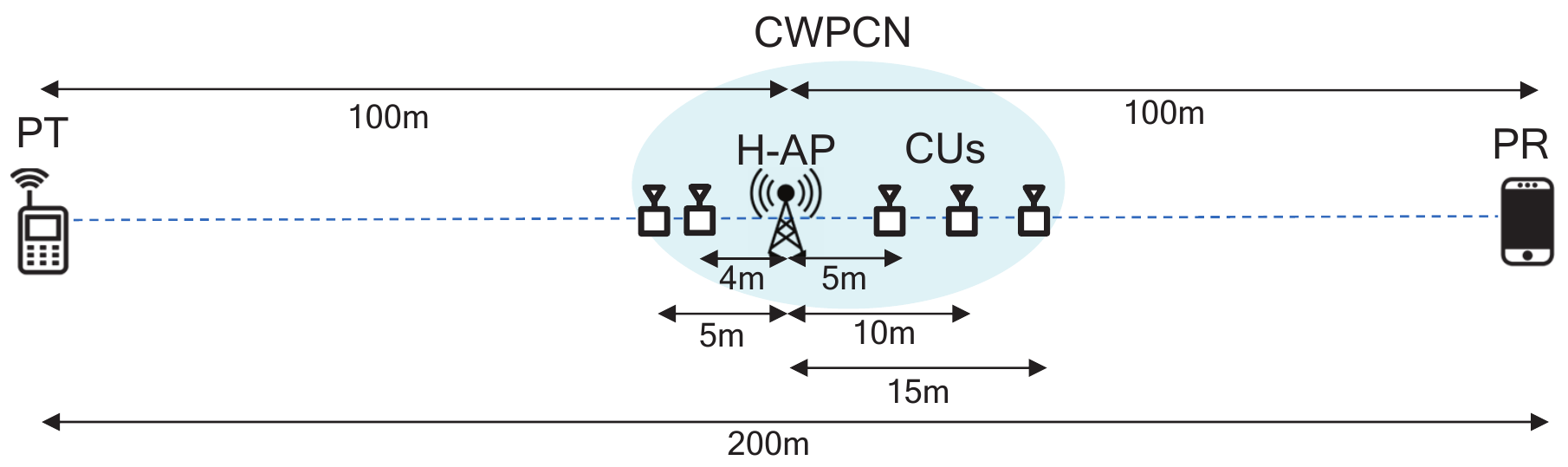} \label{Fig:Middle}}
\subfigure[Case 2]{
\centering
\includegraphics[width=8.5cm]{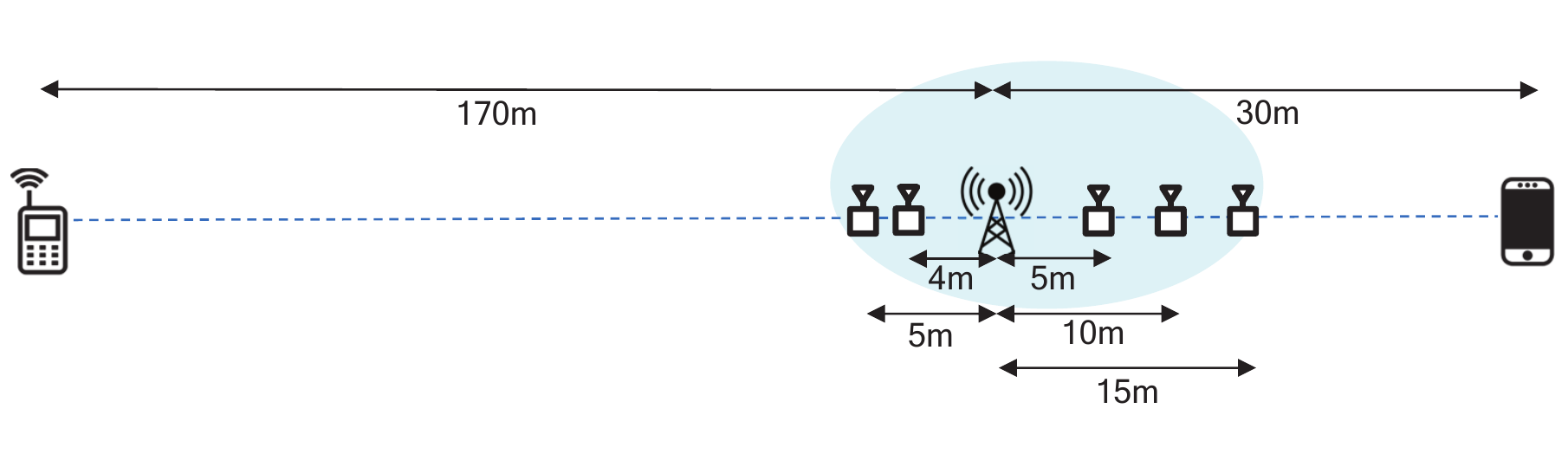} \label{Fig:PR}}
\subfigure[Case 3]{
\centering
\includegraphics[width=8.5cm]{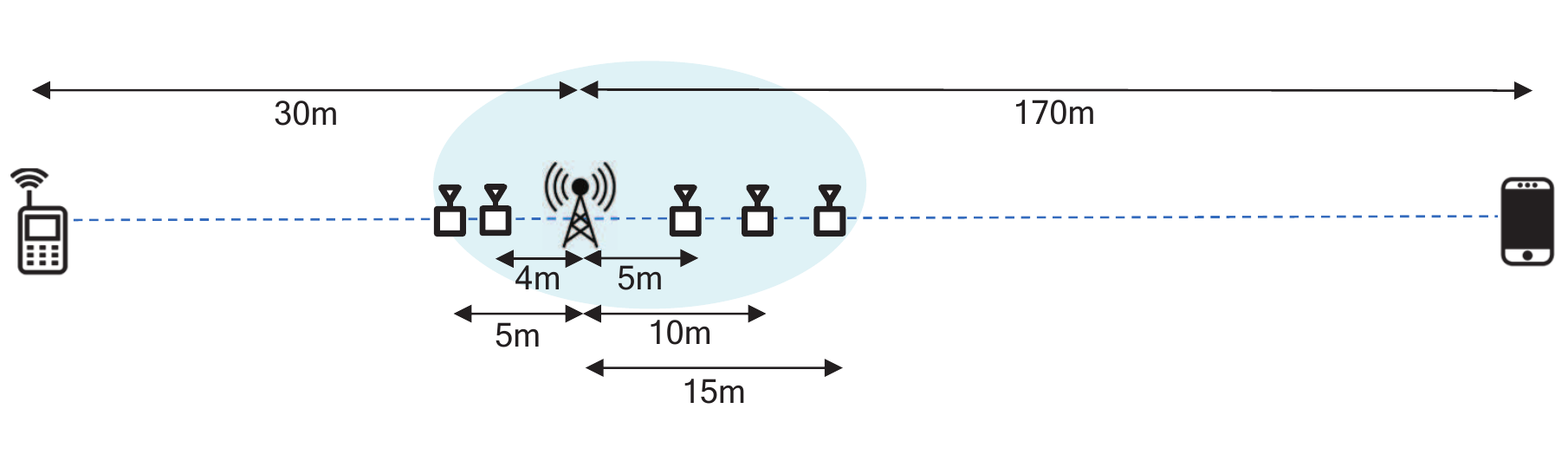}\label{Fig:PT}}
\caption{Simulation setup.} \label{Fig:SimSetup} 
\end{figure}

\begin{figure}
\centering
\subfigure[$P_p = 0.1$W]{
\centering
\includegraphics[width=8.5cm]{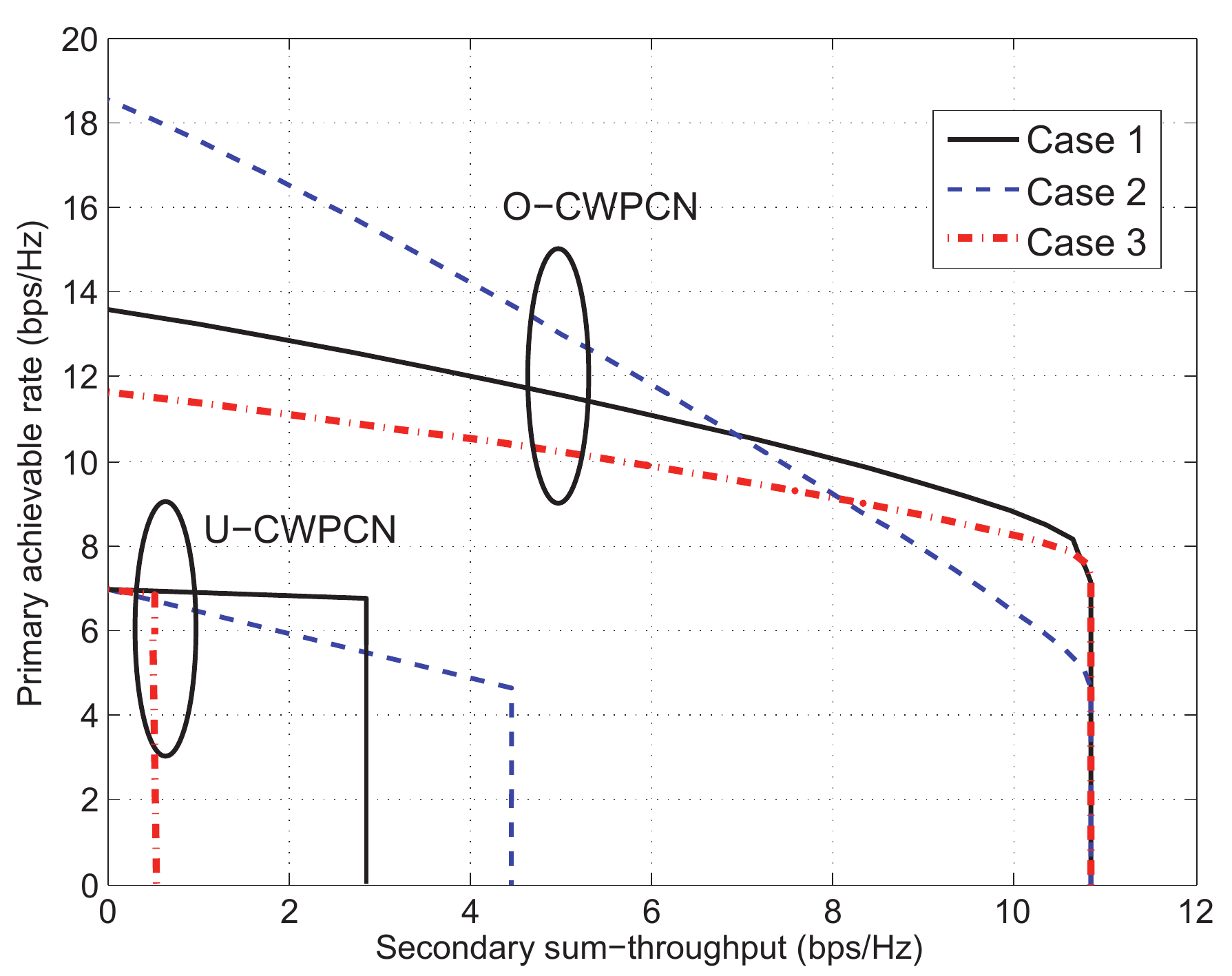} \label{Fig:RateRegion1}}
\subfigure[$P_p = 1$W]{
\centering
\includegraphics[width=8.5cm]{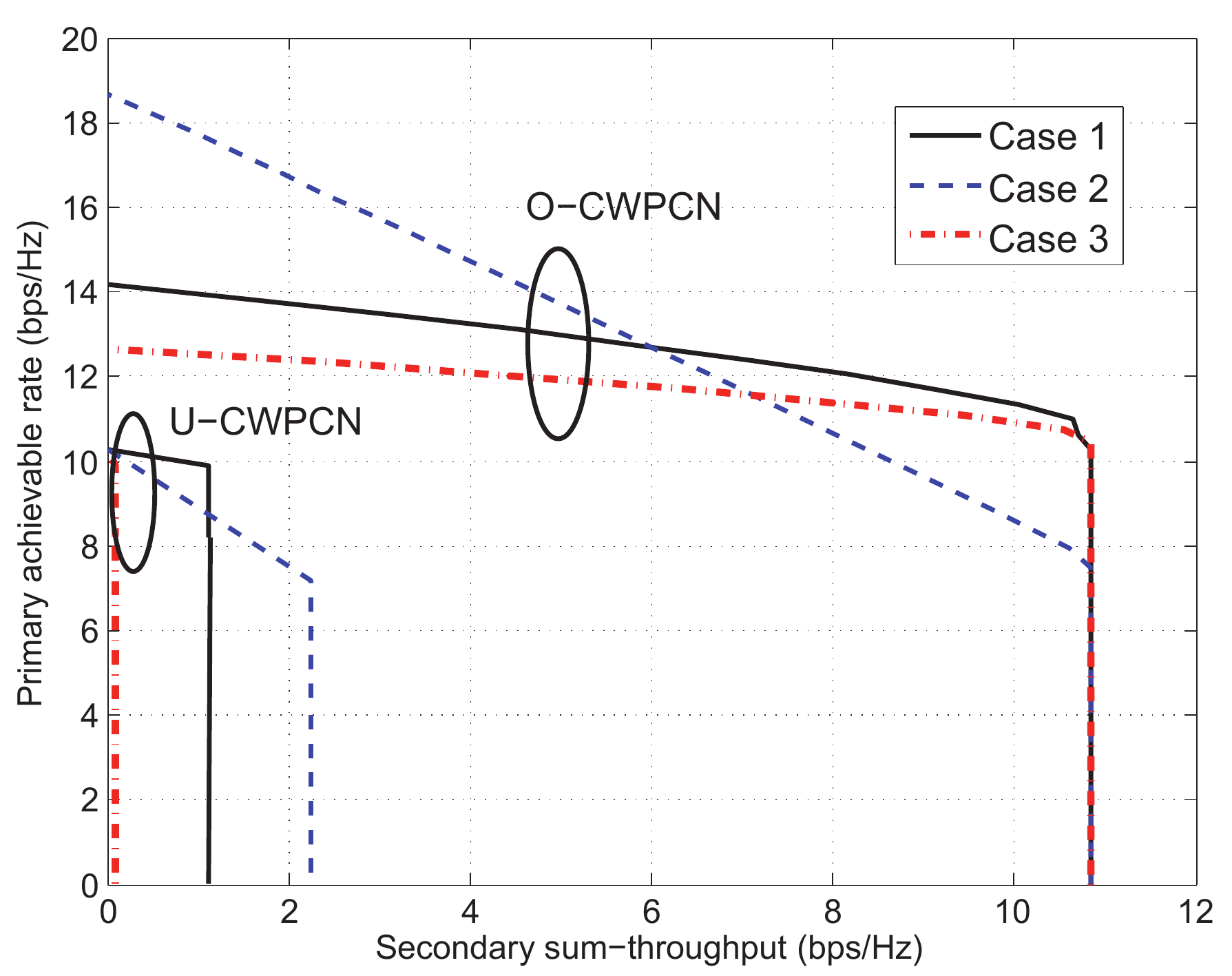} \label{Fig:RateRegion2}}
\caption{Achievable rate regions for U-CWPCN versus O-CWPCN with different PT transmit power $P_p$.} \label{Fig:RateRegion} 
\end{figure}

Under the above setup, the achievable rate regions for U-CWPCN and O-CWPCN defined in \eqref{Eq:RateRegionU} and \eqref{Eq:RateRegionO}, respectively, are shown in Figs.~\ref{Fig:RateRegion1} and \ref{Fig:RateRegion2} with the transmit power of the PT $P_p = 0.1$ Watt (W) and $P_p = 1$W, respectively, where the channel  is assumed to be subject to path-loss attenuation only without fading. Note that we set the maximum transmit power of the H-AP as $\Pmax = 1$W, the path-loss exponent to be $3$, and $20$ dB average signal power attenuation at a reference distance of $1$m. First, it can be observed from both Figs.~\ref{Fig:RateRegion1} and \ref{Fig:RateRegion2} that, for each case, the rate region for O-CWPCN is no smaller than that for U-CWPCN, which is in accordance to Proposition~\ref{Proposition:Performance}. Second, for the O-CWPCN, it can be seen from both the figures that Case 2 leads to the largest primary achievable rate when the secondary sum-throughput is zero, due to the most effective cooperative transmission by the H-AP in closest proximity. Third, for the U-CWPCN, it is observed from both the figures that Case 2 achieves the largest secondary sum-throughput when the primary rate is at minimum, due to the weakest interference from the PT to the H-AP. Last, comparing Figs.~\ref{Fig:RateRegion1} and \ref{Fig:RateRegion2}, it can be seen that the maximum secondary sum-throughput of the U-CWPCN with $P_p=0.1$W is higher than that with $P_p = 1$W for all cases due to the stronger interference from the PT when $P_p$ is larger, whereas the maximum sum-throughput of the O-CWPCN is almost unaffected by $P_p$ since the H-AP is able to cancel the interference from the PT.

Next, in Figs.~\ref{Fig:Rc_Pmax} and \ref{Fig:Rc_Alpha}, we compare the average secondary sum-throughput over flat-fading channels for different node placement cases based on the setup in Fig.~\ref{Fig:SimSetup}, by averaging over randomly generated fading channels. For the channel model, we assume both path-loss attenuation and Rayleigh fading, where a channel power gain $h$ between any two terminals can be expressed as
\begin{equation*}
h = \underline{h}\, c_0 \l(\frac{r}{r_0}\r)^{-\alpha},
\end{equation*} 
where $\underline{h}$ is an exponential random variable with unit mean denoting the Rayleigh fading; $c_0 = -20$ dB is a constant attenuation due to the path-loss at a reference distance $r_0=1$ m; $\alpha$ is the path-loss exponent, and $r$ is the distance between the associated terminals. Note that we set $\Gamma = -60$ dBm for the U-CWPCN (see (P1)), and $\bar{R}=5$ for the O-CWPCN (see (P2)). Furthermore, we set $P_p=0.1$W for Figs.~\ref{Fig:Rc_Pmax} and \ref{Fig:Rc_Alpha}.

Under this setup, Fig.~\ref{Fig:Rc_Pmax} shows the average sum-throughput comparison in different cases for the U-CWPCN and O-CWPCN, respectively, with respect to the maximum H-AP transmit power $\Pmax$. Note that we set $\alpha = 3$ for Fig.~\ref{Fig:Rc_Pmax}. It is observed from Fig.~\ref{Fig:Rc_Pmax} that for both U-CWPCN and O-CWPCN, the average sum-throughput is non-decreasing as $\Pmax$ increases. However, as can be seen from the results, the sum-throughput of the U-CWPCN eventually becomes saturated for all cases as $\Pmax$ becomes sufficiently large, since the H-AP and CUs' transmissions are constrained by the ITC, whereas that of the O-CWPCN can keep increasing for all cases since the H-AP's transmission does not interfere with but instead cooperatively assists the primary transmission, which becomes more effective as $\Pmax$ increases. Next, it can be seen from Fig.~\ref{Fig:Rc_Pmax_Underlay} that for the U-CWPCN, Case 1 shows the best performance when $\Pmax$ is large, which can be explained as follows. When the H-AP is closer to the PR as in Case 2, although the interference from the PT is weaker, the H-AP's DL WET and the CUs' UL WIT are severely constrained by the ITC imposed at the PR, while when the H-AP is closer to the PT as in Case 3, the dominant interference from the PT degrades the secondary sum-throughput significantly. Last, it can be observed from Fig.~\ref{Fig:Rc_Pmax_Overlay} that for the O-CWPCN, Case 3 achieves the best sum-throughput when $\Pmax$ is small, whereas Case 2 shows the worst performance when $\Pmax$ is large. This is because when $\Pmax$ is small, the CWPCN needs to rely on harvesting energy from the PT's transmission which is more effective when the CUs are closer to the PT as in Case 3; however, when $\Pmax$ is large, the interference caused by CUs' UL WIT to the PR becomes more severe if they are closer to the PR as in Case 2, which in turn limits the CUs' transmit power for WIT due to the PRC.

\begin{figure}
\centering
\subfigure[U-CWPCN]{
\centering
\includegraphics[width=8.5cm]{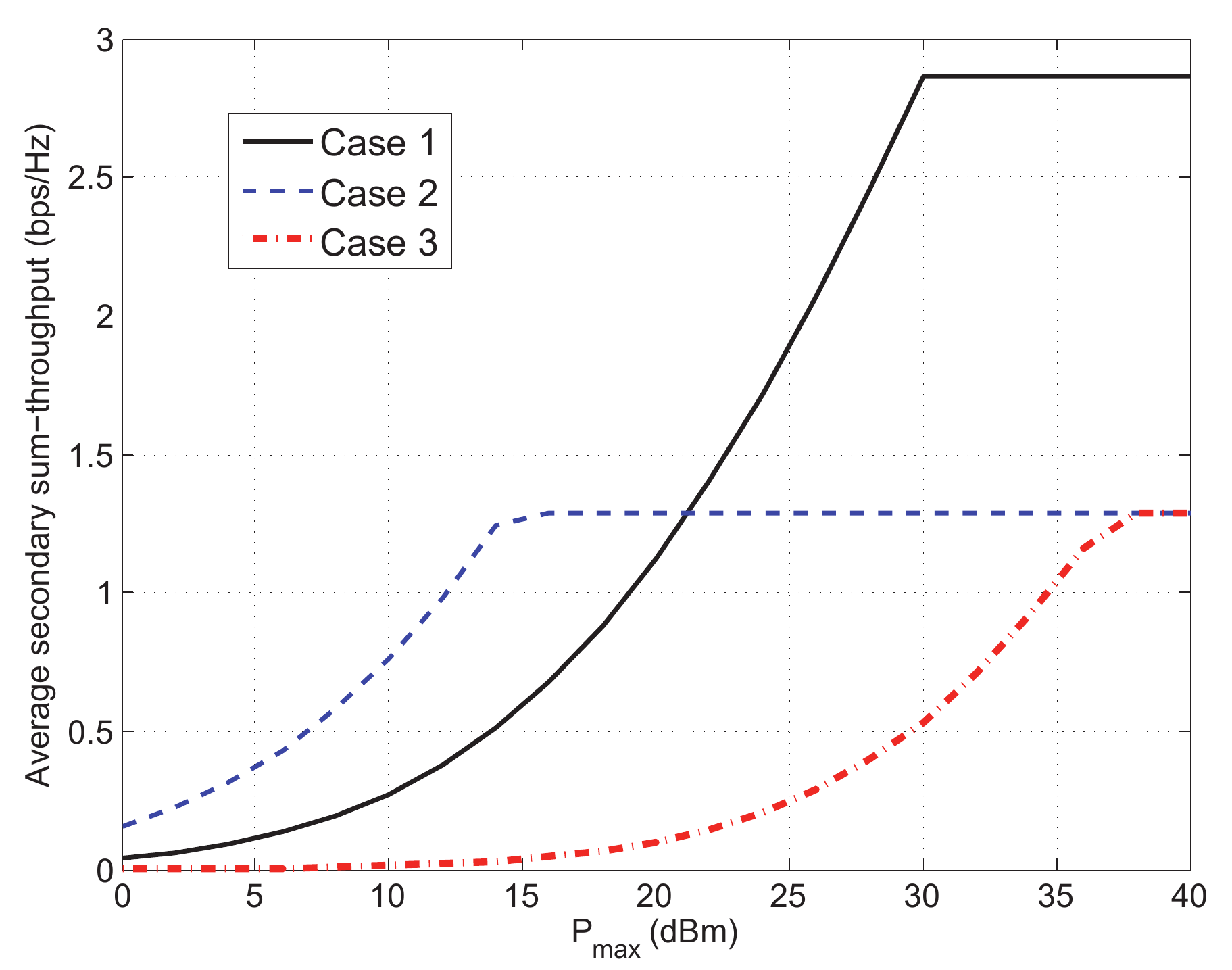} \label{Fig:Rc_Pmax_Underlay}}
\subfigure[O-CWPCN]{
\centering
\includegraphics[width=8.5cm]{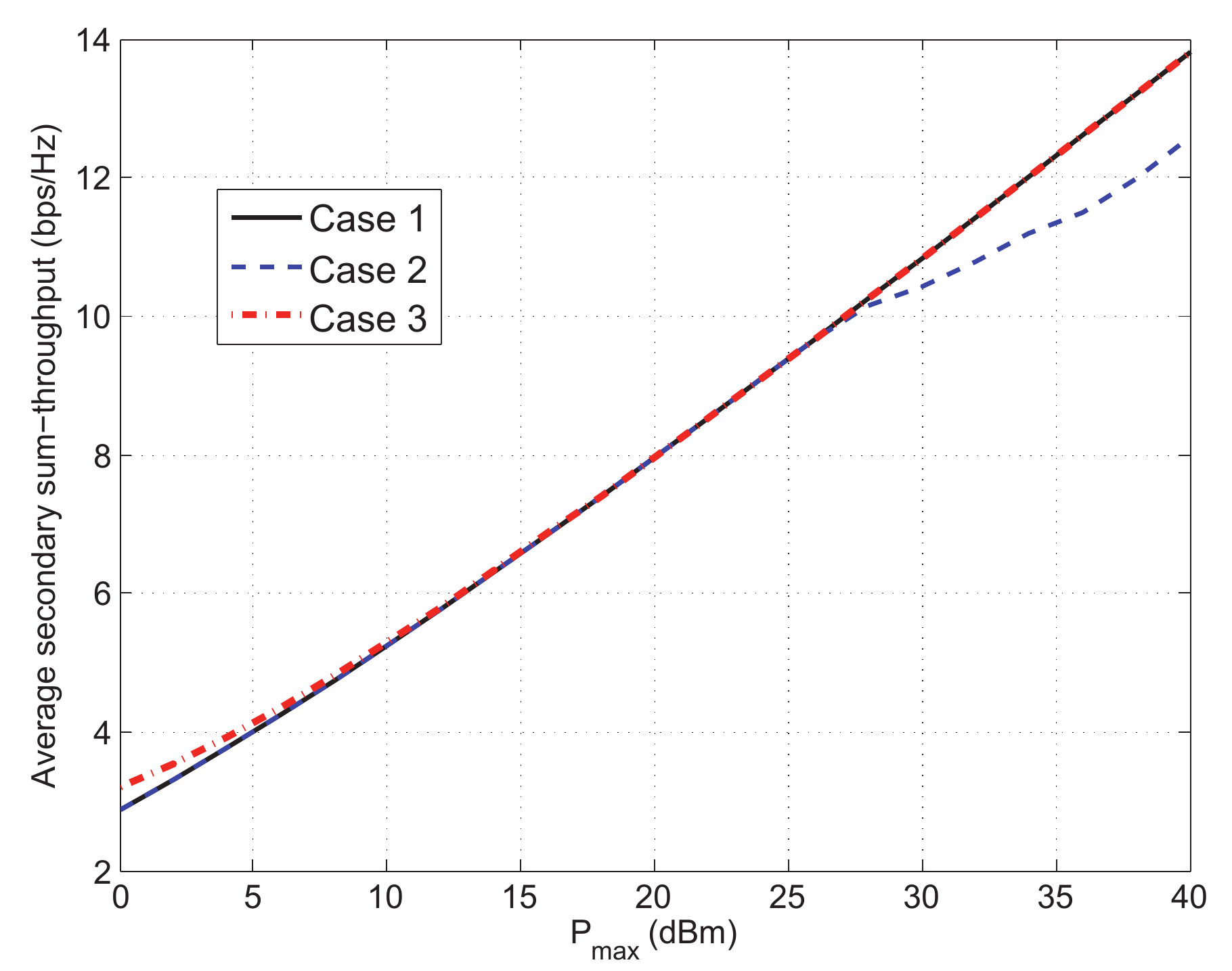} \label{Fig:Rc_Pmax_Overlay}}
\caption{Average secondary sum-throughput versus the maximum H-AP transmit power $\Pmax$.} \label{Fig:Rc_Pmax} 
\end{figure}

In Fig.~\ref{Fig:Rc_Alpha}, we plot the average sum-throughput of the U-CWPCN and O-CWPCN in different node placement cases, versus the path-loss exponent $\alpha$. Note that we set $\Pmax = 1$W for Fig.~\ref{Fig:Rc_Alpha}. First, it can be observed from Fig.~\ref{Fig:Rc_Alpha_Underlay} that, for all cases, as $\alpha$ increases, the average sum-throughput of the U-CWPCN first increases when $\alpha$ is small but finally decreases when $\alpha$ is sufficiently large. This is in sharp contrast to the result in the case of a single stand-alone WPCN \cite{J_JZ:2014}, where more significant path-loss attenuation is always harmful to both DL WET and UL WIT and thus results in decreased sum-throughput. In particular, the result in Fig.~\ref{Fig:Rc_Alpha_Underlay} implies that, in spectrum sharing based CWPCNs, large path-loss could be helpful to a certain extent  since it can help reduce the effect of interference to/from the H-AP. However, in contrast, for the O-CWPCN as shown in Fig.~\ref{Fig:Rc_Alpha_Overlay}, the secondary sum-throughput decreases for all cases as $\alpha$ becomes large, similar to that in a single WPCN, since the H-AP does not cause/receive any interference to/from the primary transmission.

\begin{figure}
\centering
\subfigure[U-CWPCN]{
\centering
\includegraphics[width=8.5cm]{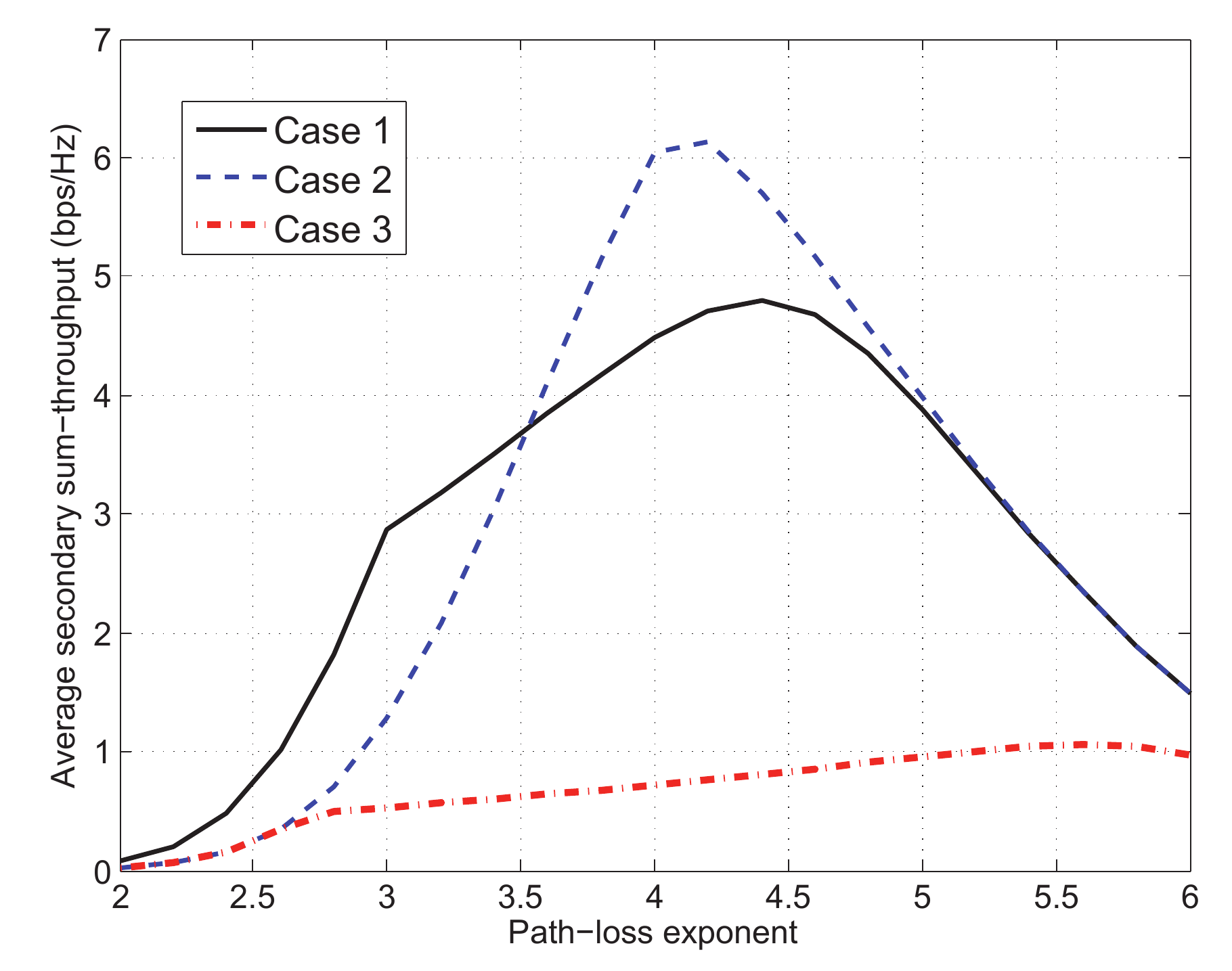} \label{Fig:Rc_Alpha_Underlay}}
\subfigure[O-CWPCN]{
\centering
\includegraphics[width=8.5cm]{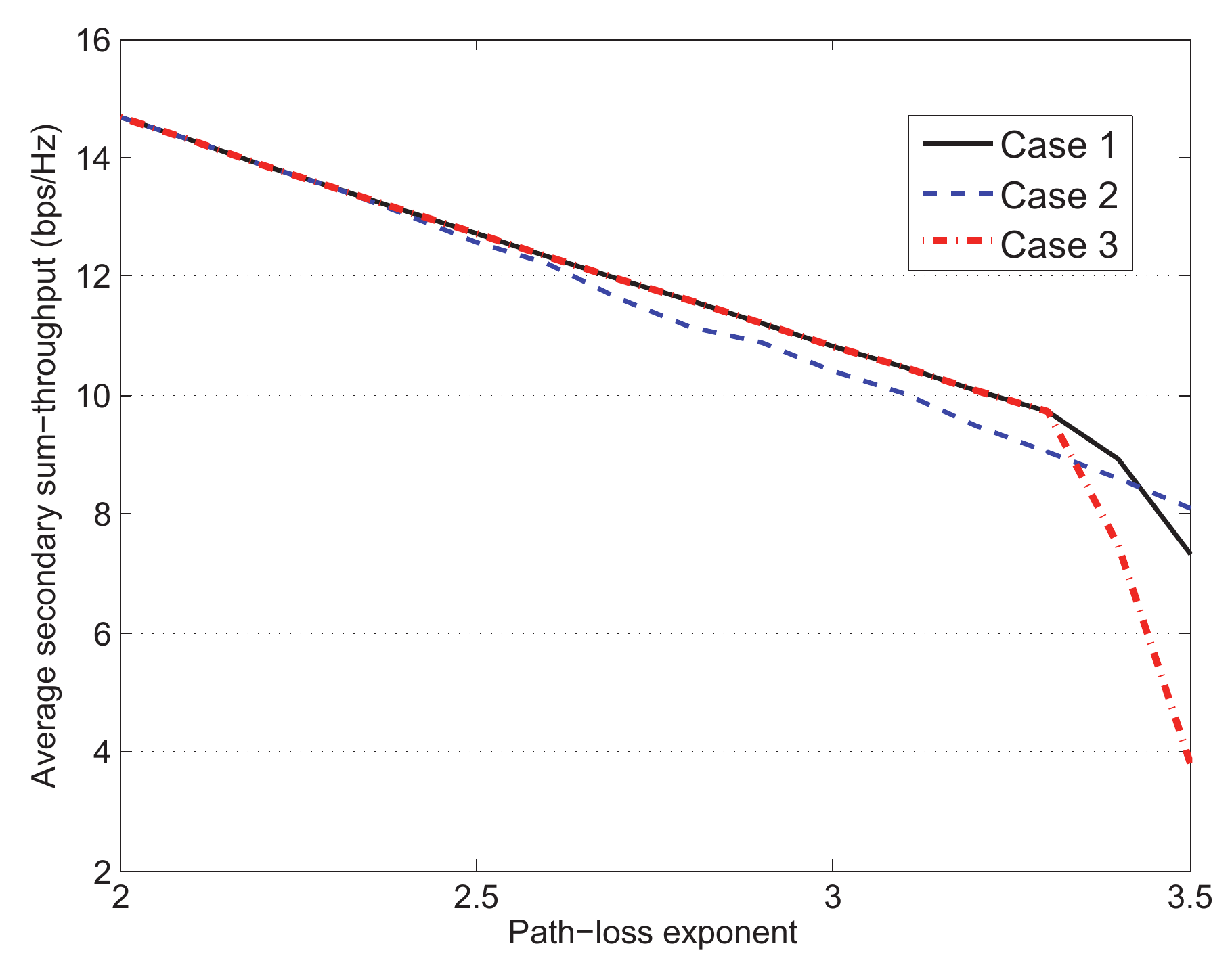} \label{Fig:Rc_Alpha_Overlay}}
\caption{Average secondary sum-throughput versus the path-loss exponent $\alpha$.} \label{Fig:Rc_Alpha} 
\end{figure}

\section{Conclusion and Future Work} \label{Section:Conclusion}
In this paper, we studied a new WPCN termed CWPCN where a secondary or cognitive WPCN shares the same spectrum for its WET and WIT with an existing primary communication system. We  proposed two spectrum sharing models of the CWPCN, namely  U-CWPCN and O-CWPCN, by extending the existing techniques for CR. We solved the sum-throughput maximization problem for the CWPCN in each model under ITC/PRC to ensure the primary communication performance, by jointly optimizing the time and power allocations in the secondary CWPCN for WET and WIT. Both analytical and simulation results were provided to compare the performance  of the CWPCNs under different models and practical setups.

There are some interesting and worth-investigating directions for future work, which are briefly discussed as follows.
\begin{itemize}
\item \emph{General CWPCN Setups:} It will be interesting to extend this paper to more general setups of CWPCN with e.g., multi-channel spectrum sharing, multi-antenna terminals and/or multiple primary links. It is noted that spectrum sharing between a primary MIMO WET system and a secondary MIMO WIT system has been recently studied in \cite{J_XBZ:2015} to minimize the performance degradation of the WIT system due to the interference caused by WET. Nevertheless, there are many interesting network models for the general wireless energy and information transmission coexisting systems not yet explored, and cognitive radio is a key technique for achieving reliable and efficient operation of such systems. 

\item \emph{User Fairness:} In this paper, we focus on sum-throughput maximization of the CWPCN, which may result in unfair rate allocation over different CUs. To overcome this problem, user fairness can be considered in problem formulations with e.g., fixed rate requirement at each CU. In this case, new issues arise such as energy outage due to different energy availability at each CU, for which the optimal transmission design of  CWPCN still remains open. 

\item \emph{Full-Duplex Systems:} The analysis in this paper is based on the assumption that all the user terminals operate in half-duplex mode. If the H-AP/CUs operate in full-duplex mode where they can transmit/harvest energy in the DL and receive/transmit information in the UL at the same time and over the same band \cite{J_JZ:2014b,J_ZZ:2015}, new transmission designs need to be considered, for which further investigation is required. 
\end{itemize}

\appendices

\section{Proof of Proposition~\ref{Proposition:Performance}} \label{Proof:Proposition:Performance}
We first define
\begin{align*}
\tilde{\cR}_O = & \bigcup_{\bar{R}\geq 0}\bigg\{(R_p, R_c): R_p \leq R_{\text{P-O}}(\tau,\{e_i\}), \nn\\ 
&\quad R_c\leq R_{\text{C}}(\tau,\{e_i\}, \{\gamma_i\}), (\tau,\{e_i\})\in \cF_O(\bar{R})\bigg\},
\end{align*}
where $\{\hat{\gamma}_i\}$ in \eqref{Eq:RateRegionO} is replaced with $\{\gamma_i\}$ (recall that $\gamma_i = h_i/(\sigma_c^2 + \hat{g}P_p)$ and $\hat{\gamma}_i = h_i/\sigma_c^2$). Since $\gamma_i \leq \hat{\gamma}_i$, $i=1,...,K$, in general, we have $R_{\text{C}}(\tau,\{e_i\}, \{\gamma_i\}) \leq R_{\text{C}}(\tau,\{e_i\}, \{\hat{\gamma}_i\})$ (see \eqref{Eq:RateCWPCN}). It thus follows that  $\tilde{\cR}_O \subseteq \cR_O$. Hence, to show that $\cR_U\subseteq\cR_O$, it suffices to prove that $\cR_U\subseteq\tilde{\cR}_O$. Suppose $(R_p,R_c)\in\cR_U$. From \eqref{Eq:RateRegionU}, it can be inferred that there must exist a $\Gamma\geq 0 $ and corresponding $(\tau,\{e_i\})\in \cF_U(\Gamma)$ that satisfy $R_p \leq R_\text{P-U}(\tau,\{e_i\})$ and $R_c \leq R_{\text{C}}(\tau,\{e_i\}, \{\gamma_i\})$. It then follows that
\begin{align}
R_p &\leq  R_\text{P-U}(\tau,\{e_i\}) \nn\\
& = \tau \log_2\l(1+\frac{g P_p}{\sigma_p^2 + \hat{h}  P_c}\r)  \nn\\
&\qquad + (1-\tau)\log_2\l(1+\frac{g P_p}{\sigma_p^2 +  \frac{1}{1-\tau}\sum_{i=1}^K \hat{h}_i e_i}\r) \nn\\
&\leq \tau r_1 + (1-\tau) \log_2\l(1+\frac{g P_p}{\sigma_p^2 +  \frac{1}{1-\tau}\sum_{i=1}^K \hat{h}_i e_i}\r) \label{Eq:UsubsetO1}\\
&=  R_{\text{P-O}}(\tau,\{e_i\}), \label{Eq:UsubsetO2}
\end{align}
where the inequality in \eqref{Eq:UsubsetO1} is due to $ \log_2\l(1+\frac{g P_p}{\sigma_p^2 + \hat{h}  P_c}\r) \leq r_1$ and $0\leq \tau \leq 1$ (see \eqref{Eq:PrimaryRate}). Next, with the given $\Gamma$ and $(\tau,\{e_i\})\in \cF_U(\Gamma)$, we set $\bar{R}>0$ such that  $0< \bar{R} \leq R_{\text{P-O}}(\tau,\{e_i\})$, with which it is ensured that $(\tau,\{e_i\})\in \cF_O(\bar{R})$ as well as $(\tau,\{e_i\})\in \cF_U(\Gamma)$. To summarize, if $(R_p,R_c)\in\cR_U$, we can always construct some $\bar{
R}$ and $(\tau,\{e_i\})\in\cF_O(\bar{R})$     such that $R_p\leq R_{\text{P-O}}(\tau,\{e_i\})$ and $R_c \leq R_{\text{C}}(\tau,\{e_i\}, \{\gamma_i\})$,  which implies that $(R_p,R_c)\in\tilde{\cR}_O$.  Thus, $\cR_U \subseteq \tilde{\cR}_O$ is always true. Consequently, we conclude that $\cR_U \subseteq \cR_O$, given $\tilde{\cR}_O \subseteq \cR_O$ as shown above. 

The proof of Proposition~\ref{Proposition:Performance} is thus completed.

\section{Proof of Proposition~\ref{Theorem:Underlay}} \label{Proof:Underlay}
It can be easily verified that (P1.1) is equivalent to the following problem.
\begin{align}
\text{(P1.1-a)}:~\mathop{\mathtt{Maximize}}_{\{e_i\}} &~~ \sum_{i=1}^K \gamma_i e_i  \nn\\
\mathtt{subject\; to}
&~~0\leq e_i \leq \bar{\tau} Q_i, \; i=1,...,K \label{P1.1'Constraint1} \\
&~~\sum_{i=1}^K \hat{h}_i e_i \leq (1-\bar{\tau})\Gamma.  \label{P1.1'Constraint2}
\end{align}
 The Lagrangian of (P1.1-a) can be expressed as
\begin{align*}
&\cL(\{e_i\}, \{\nu_i\}, \{\mu_i\}, \lambda) = \sum_{i=1}^K \gamma_i e_i + \sum_{i=1}^K \nu_i e_i  \nn\\ 
&\quad\qquad - \sum_{i=1}^K \mu_i(e_i - \bar{\tau} Q_i)  - \lambda\l(\sum_{i=1}^K \hat{h}_i e_i - \Gamma(1-\bar{\tau})\r),
\end{align*}
where $(\mu_i, \nu_i)$, $i=1,...,K$, and $\lambda$ are the non-negative dual variables associated with the constraints \eqref{P1.1'Constraint1} and \eqref{P1.1'Constraint2}, respectively. Because (P1.1-a) is convex, the optimal primal and dual solutions must satisfy the KKT conditions given as follows  \cite{B_BV:2004}. 
\begin{align}
\nu_i^\star \bar{e}_i & = 0, \quad i=1,...,K \label{Eq:KKT5}\\
\mu_i^\star(\bar{e}_i - \bar{\tau} Q_i) &= 0, \quad i=1,...,K \label{Eq:KKT1}\\
\lambda^\star\l(\sum_{i=1}^K \hat{h}_i \bar{e}_i - \Gamma(1-\bar{\tau})\r) & = 0, \label{Eq:KKT2}\\
\gamma_i +\nu_i^\star - \mu_i^\star - \lambda^\star \hat{h}_i &= 0, \quad i=1,...,K  \label{Eq:KKT3}
\end{align}
where  $\mu_i^\star$, $\nu_i^\star$, $i=1,...,K$, and $\lambda^\star$ denote the optimal dual variables, and $\{\bar{e}_i\}$ is the optimal solution to (P1.1-a). First, consider the case of $\lambda^\star = 0$, which corresponds to item 1) of Proposition~\ref{Theorem:Underlay}.  From \eqref{Eq:KKT3}, we obtain $\mu_i^\star = \gamma_i + \nu_i^\star > 0$, $i=1,...,K$. Thereby, from \eqref{Eq:KKT1}, it follows that $\bar{e}_i = \bar{\tau} Q_i$, $i=1,..,K$. Since  $\bar{e}_i$'s must satisfy \eqref{P1.1'Constraint2}, we have $\sum_{i=1}^K \hat{h}_i \bar{\tau} Q_i \leq \Gamma(1-\bar{\tau})$, i.e., $0 < \bar{\tau}\leq \frac{\Gamma}{\Gamma + \sum_{i=1}^K}\hat{h}_i Q_i$. This completes the  ``only if'' part of item 1) of Proposition~\ref{Theorem:Underlay}. For the ``if'' part, suppose we set $0 < \bar{\tau}\leq \frac{\Gamma}{\Gamma + \sum_{i=1}^K}\hat{h}_i Q_i$, i.e.,  $\sum_{i=1}^K \hat{h}_i \bar{\tau} Q_i \leq \Gamma(1-\bar{\tau})$. In this case, it can be verified that the choice of $e_i = \bar{\tau} Q_i$, $i=1,...,K$, is feasible as well as optimal for (P1.1-a), i.e., we obain $\bar{e}_i = \bar{\tau} Q_i$, $i=1,...,K$. To summarize, $\bar{e}_i = \bar{\tau} Q_i$, $i=1,...,K$, if and only if $0<\bar{\tau}\leq \frac{\Gamma}{\Gamma + \sum_{i=1}^K}\hat{h}_i Q_i$. 

Next, consider the case of $\lambda^\star > 0$, which corresponds to item 2) of Proposition~\ref{Theorem:Underlay}. In this case, we first show that at most one CU transmits with fractional power, i.e., $0 < \bar{e}_i < \bar{\tau} Q_i$ for some $i$, and the other CUs either transmit with full power or do not transmit at all. We show this by contradiction. Suppose that two CUs, say CU$_m$ and CU$_n$, transmit with fractional power, i.e.,  $0< \bar{e}_m< \bar{\tau} Q_m$ and $0< \bar{e}_n< \bar{\tau} Q_n$. From \eqref{Eq:KKT5} and \eqref{Eq:KKT1}, this results in $\nu_m^\star=\nu_n^\star=0$ and $\mu_m^\star=\mu_n^\star=0$. Thus, from \eqref{Eq:KKT3}, we obtain
\begin{equation*}
\lambda^\star = \frac{\gamma_m}{\hat{h}_m} = \frac{\gamma_n}{\hat{h}_n}.
\end{equation*}
Note that this event occurs with zero probability due to the independence of the continuous random variables $\gamma_m$, $\gamma_n$, $\hat{h}_m$, and $\hat{h}_n$. Therefore, the presumption that two CUs transmit both with fractional power cannot be true. This can be easily extended to the case where more than two CUs transmit with fractional power, and thus we conclude that there is at most one CU$_i$ with $0 < \bar{e}_i < \bar{\tau}Q_i$. Now, suppose CU$_m$ is the one that transmits with fractional power, i.e., $0 < \bar{e}_m< \bar{\tau} Q_m$. From \eqref{Eq:KKT5} and \eqref{Eq:KKT1}, it follows that $\nu_m^\star=0$ and $\mu_m^\star=0$. In the following, we show that if $\frac{\gamma_l}{\hat{h}_l}>\frac{\gamma_m}{\hat{h}_m}$, $l\neq m$, then we have $\bar{e}_l= \bar{\tau} Q_l$, and if $\frac{\gamma_l}{\hat{h}_l}<\frac{\gamma_m}{\hat{h}_m}$, $l\neq m$, then we have $\bar{e}_l=0$. First, suppose that $\frac{\gamma_l}{\hat{h}_l}>\frac{\gamma_m}{\hat{h}_m}$, $l\neq m$. From \eqref{Eq:KKT3}, it follows that
\begin{equation*}
\lambda^\star = \frac{\gamma_l + \nu_l^\star - \mu_l^\star}{\hat{h}_l} =  \frac{\gamma_m}{\hat{h}_m},
\end{equation*}
with which it must follow that $\mu_l^\star>\nu_l^\star\geq 0$  to satisfy $\frac{\gamma_l}{\hat{h}_l}>\frac{\gamma_m}{\hat{h}_m}$. From \eqref{Eq:KKT1}, this results in $\bar{e}_l = \bar{\tau} Q_l$. Similarly, when $\frac{\gamma_l}{\hat{h}_l}<\frac{\gamma_m}{\hat{h}_m}$, $l\neq m$, we can easily show that $\bar{e}_l=0$. Therefore, we conclude that if $m=(k)$, $1\leq k \leq K$, then $\bar{e}_{(i)}=\bar{\tau} Q_{(i)}$ if $i=1,...,k-1$, and $\bar{e}_{(i)} = 0$ if $i=k+1,...,K$, where the set of indices $\{(1),...,(K)\}$ has been defined in the proposition. Since we have $\lambda^\star > 0$, from \eqref{Eq:KKT2}, it must follow that $\sum_{i=1}^K \hat{h}_i \bar{e}_i = \Gamma(1-\bar{\tau})$, i.e., the interference constraint \eqref{P1.1Constraint2} must be tight.  Thus, using  $\bar{e}_{(i)}=\tau Q_{(i)}$, $i=1,...,k-1$, and $\bar{e}_{(i)} = 0$, $i=k+1,...,K$, we obtain
\begin{equation} \label{Eq:UnderlayFractional}
\bar{e}_{(k)} = \frac{1}{\hat{h}_{(k)}}\l(\Gamma(1-\bar{\tau}) - \bar{\tau}\sum_{i=1}^{k-1}\hat{h}_{(i)}Q_{(i)}\r).
\end{equation}
With \eqref{Eq:UnderlayFractional} and the fact that $0 < \bar{e}_{(k)} < \bar{\tau} Q_{(k)}$, we have $\frac{\Gamma}{\Gamma + \sum_{i=1}^{k} \hat{h}_{(i)} Q_{(i)}} < \bar{\tau} < \frac{\Gamma}{\Gamma + \sum_{i=1}^{k-1} \hat{h}_{(i)} Q_{(i)}}$. Moreover, it is easy to see that when $\bar{\tau} = \frac{\Gamma}{\Gamma + \sum_{i=1}^{k-1} \hat{h}_{(i)} Q_{(i)}}$, we obtain $\bar{e}_{(i)} = \bar{\tau} Q_{(i)}$, $i=1,...,k-1$, and $\bar{e}_{(i)} = 0$, $i=k,k+1,...,K$, i.e., each CU either transmits with full power or zero power. This shows the ``only if'' part of item 2) of Proposition~\ref{Theorem:Underlay}. Next, consider the ``if'' part of item 2) of Proposition~\ref{Theorem:Underlay}. Suppose we set $\frac{\Gamma}{\Gamma + \sum_{i=1}^{k} \hat{h}_{(i)} Q_{(i)}} < \bar{\tau} \leq \frac{\Gamma}{\Gamma + \sum_{i=1}^{k-1} \hat{h}_{(i)} Q_{(i)}}$, i.e., $\frac{\bar{\tau}}{1-\bar{\tau}}\sum_{i=1}^{k-1}\hat{h}_{(i)}Q_{(i)} \leq \Gamma < \frac{\bar{\tau}}{1-\bar{\tau}}\sum_{i=1}^{k}\hat{h}_{(i)}Q_{(i)}$. Due to the fact that at most one CU transmits with fractional power, it is easy to see that the choice of $\bar{e}_{(i)} = \bar{\tau}Q_{(i)}$, $i=1,...,k-1$, $\bar{e}_{(k)} = \frac{1}{\hat{h}_{(k)}}\l(\Gamma(1-\bar{\tau}) - \bar{\tau}\sum_{i=1}^{k-1}\hat{h}_{(i)}Q_{(i)}\r)$, and $\bar{e}_{(i)}=0$, $i=k+1,...,K$, is feasible as well as optimal for (P1.1-a). 

The proof of Proposition~\ref{Theorem:Underlay} is thus completed. 

\section{Proof of Proposition~\ref{Proposition:ActiveInterference}} \label{Proof:Proposition:ActiveInterference}
We first show that $R^\star \geq R_\text{a}^\star(\tilde{\Gamma}_0)$. It can be easily verified that for any given $\Gamma_0\geq 0$, if $(\tau, \{e_i\})$ is feasible for (P2-a), then it is always feasible for (P2), but not vice versa. This implies that the optimal value of (P2-a) is in general a lower bound of that of (P2), i.e., $R^\star\geq R_\text{a}^\star(\Gamma_0)$, $\forall \Gamma_0\geq 0$. This gives $R^\star \geq R_\text{a}^\star(\tilde{\Gamma}_0)$. 

Next, we show that $R^\star\leq R_\text{a}^\star(\tilde{\Gamma}_0)$ is also true. Let $(\tau^\star,\{e_i^\star\})$ denote the optimal solution of (P2), and define $\Gamma^\star = \frac{1}{1-\tau^\star}\sum_{i=1}^K \hat{h}_i e_i^\star$, i.e., $\Gamma^\star$ is the resulting interference-temperature level at the PR with the optimal solution of (P2). With $\Gamma_0 = \Gamma^\star$ in (P2-a), it can be easily verified that $(\tau^\star,\{e_i^\star\})$ is a  feasible solution of (P2-a), i.e., $R_\text{a}^\star(\Gamma^\star) \geq R^\star$. Since $R_\text{a}^\star(\tilde{\Gamma}_0) \geq R_\text{a}^\star(\Gamma_0)$, $\forall \Gamma_0\geq 0$, it follows that $R_\text{a}^\star(\tilde{\Gamma}_0) \geq R_\text{a}^\star(\Gamma^\star) \geq R^\star$. 

To summarize, we have shown that both $R^\star \geq R_\text{a}^\star(\tilde{\Gamma}_0)$ and $R^\star \leq R_\text{a}^\star(\tilde{\Gamma}_0)$ are true, and thus we conclude that $R^\star = R_\text{a}^\star(\tilde{\Gamma}_0)$. The proof is thus completed.

\section{Proof of Proposition~\ref{Proposition:Overlay}} \label{Proof:Proposition:Overlay}

It can be verified from Proposition~\ref{Theorem:Underlay} that the constraint \eqref{Eq:P2.1Constraint4}  is tight, i.e., $\frac{1}{1-\tau(\Gamma_0)}\sum_{i=1}^K \hat{h}_ie_i(\Gamma_0) = \Gamma_0$, if and only if $\tau(\Gamma_0) \geq \frac{\Gamma_0}{\Gamma_0 + \sum_{i=1}^K \hat{h}_i Q_i}$.  Thus, to prove Proposition~\ref{Proposition:Overlay}, it suffices to show that  $\tau(\Gamma_0) \geq \frac{\Gamma_0}{\Gamma_0 + \sum_{i=1}^K \hat{h}_i Q_i}$ if and only if $\Gamma_0 \in \cF_{\Gamma_0}$, where $\cF_{\Gamma_0}$ is defined in \eqref{Eq:GammaTightSet}. For the purpose of exposition, we present the following corollary, which is a direct consequence of Proposition~\ref{Theorem:Underlay}. 
\begin{corollary} \label{Corollary:Underlay}
We have $\hat{\tau}^\star \geq \frac{\Gamma_0}{\Gamma_0 + \sum_{i=1}^K\hat{h}_i Q_i}$ if and only if $\tilde{\tau}_K \geq \frac{\Gamma_0}{\Gamma_0 + \sum_{i=1}^K\hat{h}_i Q_i}$.
\end{corollary}
Note that $\tilde{\tau}_K$ in the above has been defined in the proposition. 

First, consider the ``if'' part. Suppose we have $\Gamma_0 \in \cF_{\Gamma_0}$, i.e., $ \tilde{\tau}_K\geq \frac{\Gamma_0}{\Gamma_0 + \sum_{i=1}^K \hat{h}_i Q_i}$ or $\frac{\bar{R} - r_2(\Gamma_0)}{r_1 - r_2(\Gamma_0)}\geq \frac{\Gamma_0}{\Gamma_0 + \sum_{i=1}^K \hat{h}_i Q_i}$. For the former case of $\tilde{\tau}_K\geq \frac{\Gamma_0}{\Gamma_0 + \sum_{i=1}^K \hat{h}_i Q_i}$, we obtain $\hat{\tau}^\star\geq \frac{\Gamma_0}{\Gamma_0 + \sum_{i=1}^K \hat{h}_i Q_i}$ from Corollary~\ref{Corollary:Underlay}; thus, from \eqref{Eq:tau_Gamma}, it must follow that $\tau(\Gamma_0) \geq \frac{\Gamma_0}{\Gamma_0 + \sum_{i=1}^K \hat{h}_i Q_i}$. For the latter case of $\frac{\bar{R}-r_2(\Gamma_0)}{r_1 - r_2(\Gamma_0)}\geq \frac{\Gamma_0}{\Gamma_0 + \sum_{i=1}^K \hat{h}_i Q_i}$, it is straightforward from \eqref{Eq:tau_Gamma} that $\tau(\Gamma_0) \geq \frac{\Gamma_0}{\Gamma_0 + \sum_{i=1}^K \hat{h}_i Q_i}$. As a result, it follows that if $\Gamma_0 \in \cF_{\Gamma_0}$, then $\tau(\Gamma_0) \geq \frac{\Gamma_0}{\Gamma_0 + \sum_{i=1}^K \hat{h}_i Q_i}$.

Next, consider the ``only if'' part. Suppose the obtained optimal $\tau$ of (P2-a) satisfies $\tau(\Gamma_0) \geq \frac{\Gamma_0}{\Gamma_0 + \sum_{i=1}^K \hat{h}_i Q_i}$. We consider two cases: $\tilde{\tau}_K < \frac{\Gamma_0}{\Gamma_0 + \sum_{i=1}^K \hat{h}_i Q_i}$ or $ \tilde{\tau}_K \geq \frac{\Gamma_0}{\Gamma_0 + \sum_{i=1}^K \hat{h}_i Q_i}$. For the former case of $\tilde{\tau}_K < \frac{\Gamma_0}{\Gamma_0 + \sum_{i=1}^K \hat{h}_i Q_i}$, it can be easily verified from \eqref{Eq:tau_Gamma} and Corollary~\ref{Corollary:Underlay} that, in order to obtain $\tau(\Gamma_0) \geq \frac{\Gamma_0}{\Gamma_0 + \sum_{i=1}^K \hat{h}_i Q_i}$, it must follow that $\frac{\bar{R}-r_2(\Gamma_0)}{r_1-r_2(\Gamma_0)}\geq \frac{\Gamma_0}{\Gamma_0 + \sum_{i=1}^K \hat{h}_i Q_i}$. For the latter case of $ \tilde{\tau}_K \geq \frac{\Gamma_0}{\Gamma_0 + \sum_{i=1}^K \hat{h}_i Q_i}$,  we can infer from Corollary~\ref{Corollary:Underlay} that we always obtain $\tau(\Gamma_0) \geq \frac{\Gamma_0}{\Gamma_0 + \sum_{i=1}^K \hat{h}_i Q_i}$. As a result, it follows that if $\tau(\Gamma_0) \geq \frac{\Gamma_0}{\Gamma_0 + \sum_{i=1}^K \hat{h}_i Q_i}$, then $\Gamma_0 \in \cF_{\Gamma_0}$. 

We have shown that $\tau(\Gamma_0) \geq \frac{\Gamma_0}{\Gamma_0 + \sum_{i=1}^K \hat{h}_i Q_i}$ if and only if $\Gamma_0 \in \cF_{\Gamma_0}$. As mentioned above, since the constraint \eqref{Eq:P2.1Constraint4} is tight if and only if $\tau(\Gamma_0) \geq \frac{\Gamma_0}{\Gamma_0 + \sum_{i=1}^K \hat{h}_i Q_i}$, the proof of Proposition~\ref{Proposition:Overlay} is thus completed.

\bibliographystyle{IEEEtran}

\newpage

\end{document}